\documentclass[11pt]{article}

\usepackage{algorithm}
\usepackage{mathrsfs}
\usepackage{enumitem}
\usepackage{bbm}
\usepackage{tikz}
\usepackage{float}
\usetikzlibrary{backgrounds}
\usepackage{color}
\usepackage{graphicx}
\usepackage{latexsym}
\usepackage{amsfonts}
\usepackage{pifont,xspace,fullpage,epsfig, wrapfig}
\usepackage{amsmath, amssymb, amsthm} 
\usepackage{multirow}
\usepackage{dsfont}
\usepackage{array}

\DeclareSymbolFont{AMSb}{U}{msb}{m}{n}
\DeclareMathSymbol{\N}{\mathbin}{AMSb}{"4E}
\DeclareMathSymbol{\Z}{\mathbin}{AMSb}{"5A}
\DeclareMathSymbol{\R}{\mathbin}{AMSb}{"52}
\DeclareMathSymbol{\Q}{\mathbin}{AMSb}{"51}
\DeclareMathSymbol{\erert}{\mathbin}{AMSb}{"50}
\DeclareMathSymbol{\I}{\mathbin}{AMSb}{"49}
\DeclareMathSymbol{\C}{\mathbin}{AMSb}{"43}

\newcommand{\mynote}[2]{{\textcolor{#1}{ #2}}}

\definecolor{gray}{gray}{0.4}
\newcommand{\gray}[1]{\mynote{gray}{{\footnotesize #1}}}

\newcommand{\remove}[1]{}

\newtheorem{theorem}{Theorem}[section]
\newtheorem{lemma}[theorem]{Lemma}
\newtheorem{definition}[theorem]{Definition}
\newtheorem{remark}[theorem]{Remark}

\newtheorem{corollary}[theorem]{Corollary}

\newtheorem{observation}[theorem]{Observation}

\newcommand{\AAA}{\mathcal A}
\newcommand{\BBB}{\mathcal B}

\newcommand{\DDD}{\mathcal D}

\newcommand{\III}{\mathcal I}
\newcommand{\JJJ}{\mathcal J}

\newcommand{\NNN}{\mathcal N}

\newcommand{\eps}{\epsilon}

\newcommand{\Lap}{\operatorname{\rm Lap}}

\newcommand{\polylog}{\mathop{\rm polylog}}
\newcommand{\poly}{\mathop{\rm poly}}
\newcommand{\tower}{\mathop{\rm tower}}

\def\Q{\operatorname*{\mathbb{Q}}}
\def\poly{\mathop{\rm{poly}}\nolimits}
\def\Lap{\mathop{\rm{Lap}}\nolimits}

\makeatletter
\newcommand{\thickhline}{%
    \noalign {\ifnum 0=`}\fi \hrule height 1pt
    \futurelet \reserved@a \@xhline
}
\newcolumntype{"}{@{\hskip\tabcolsep\vrule width 1pt\hskip\tabcolsep}}
\makeatother

\newcommand{\specialcell}[2][c]{%
  \begin{tabular}[#1]{@{}c@{}}#2\end{tabular}}

\begin{document}

\begin{titlepage}
 
\title{Locating a Small Cluster Privately{\let\thefootnote\relax\footnotetext{\hspace{-1.8em}DOI: http://dx.doi.org/10.1145/2902251.2902296\vspace{5px}}}\thanks{
The presentation and analysis of Algorithm \texttt{GoodCenter} (Section~\ref{sec:GoodCenterAnalysis}) includes a step that was missing in earlier versions of this paper.
}}

\author{
\makebox[1.5in]{\hfill Kobbi Nissim\thanks{Dept.\ of Computer Science, Ben-Gurion University {\em and} CRCS, Harvard University. Supported by NSF grant CNS-1237235, a grant from the Sloan Foundation, a Simons Investigator grant to Salil Vadhan, and ISF grant 276/12.} \hfill}\\
{\small Ben-Gurion University} \\
{\small {\em and} Harvard University}\\
{\small \tt kobbi@cs.bgu.ac.il}
\and \makebox[1.5in]{\hfill Uri Stemmer\thanks{Dept.\ of Computer Science, Ben-Gurion University. Supported by the Ministry of Science and Technology, Israel.}\hfill}\\
{\small Ben-Gurion University}\\
{\small \tt stemmer@cs.bgu.ac.il}
\and \makebox[1.5in]{\hfill Salil Vadhan\thanks{Center for Research on Computation \& Society, John A. Paulson School of Engineering \& Applied Sciences, Harvard University. \texttt{http://seas.harvard.edu/\textasciitilde salil}.
Supported by NSF grant CNS-1237235, a grant from the Sloan Foundation, and a Simons Investigator grant.  Work done in part when visiting the Shing-Tung Yau Center and the Department of Applied Mathematics at National Chiao-Tung University in Hsinchu, Taiwan.}\hfill}\\
{\small Harvard University}\\
{\small \tt salil@seas.harvard.edu}
}

\remove{

\numberofauthors{3} 
\author{
\alignauthor
Kobbi Nissim\thanks{Dept.\ of Computer Science, Ben-Gurion University {\em and} CRCS, Harvard University. Supported by NSF grant CNS-1237235, a grant from the Sloan Foundation, a Simons Investigator grant to Salil Vadhan, and ISF grant 276/12.}\\
       \affaddr{Ben-Gurion University {\em and} Harvard University}\\
       \email{kobbi@cs.bgu.ac.il}
\alignauthor
Uri Stemmer\titlenote{Dept.\ of Computer Science, Ben-Gurion University. Supported by the Ministry of Science and Technology (Israel), and by the Frankel Center for Computer Science.}\\
       \affaddr{Ben-Gurion University}\\
       \email{stemmer@cs.bgu.ac.il}
\alignauthor Salil Vadhan\titlenote{Center for Research on Computation \& Society, John A. Paulson School of Engineering \& Applied Sciences, Harvard University. \texttt{http://seas.harvard.edu/\textasciitilde salil}.
Supported by NSF grant CNS-1237235, a grant from the Sloan Foundation, and a Simons Investigator grant.  Work done in part when visiting the Shing-Tung Yau Center and the Department of Applied Mathematics at National Chiao-Tung University in Hsinchu, Taiwan.}\\
       \affaddr{Harvard University}\\
       \email{salil@seas.harvard.edu}
}

}

\date{\today}
\maketitle
\setcounter{page}{0} \thispagestyle{empty}

\begin{abstract}
We present a new algorithm for locating a small cluster of points with differential privacy~[Dwork, McSherry, Nissim, and Smith, 2006]. Our algorithm has implications to private data exploration, clustering, and removal of outliers. Furthermore, we use it to significantly relax the requirements of the sample and aggregate technique~[Nissim, Raskhodnikova, and Smith, 2007], which allows compiling of ``off the shelf'' (non-private) analyses into analyses that preserve differential privacy.
\end{abstract}

\end{titlepage}

\section{Introduction} 
Clustering -- the task of grouping data points by their similarity -- is one of the most commonly used techniques for exploring data, for identifying structure in uncategorized data, and for performing a variety of machine learning and optimization tasks. 
We present a new differentially private algorithm for a clustering-related task:
Given a collection $S$ of $n$ points in the $d$-dimensional Euclidean space $\R^d$ and a parameter $t$ reflecting a target number of points, our goal is to find a smallest ball containing at least $t$ of the input points, while preserving differential privacy.

\subsection{Problem and motivation}

We recall the definition of differential privacy.  We think of a dataset as consisting of $n$ rows from a data universe $U$, where each row corresponds to one individual.  Differential privacy requires that no individual's data has a significant effect on the distribution of what we output.

\begin{definition}
A randomized algorithm $M : U^n\rightarrow Y$ is $(\eps,\delta)$ {\em differentially private} if for every two datasets $S,S'\in U^n$ that differ on one row, and every set $T\subseteq Y$, we have
$$\Pr[M(S)\in T]\leq e^{\eps}\cdot \Pr[M(S')\in T]+\delta.$$
\end{definition}

The common setting of parameters is to take $\eps$ to be a small constant and $\delta$ to be negligible in $n$, e.g., $\delta=1/n^{\log n}$. For the introduction, we will assume that that $\delta<1/(nd)$.

In this work we study the following problem under differential privacy:

\begin{definition}
A {\em 1-cluster problem} $\mathscr{C}=(X^d,n,t)$ consists of a $d$-dimensional domain $X^d$ (where $X\subseteq\R$ is finite and totally ordered), and parameters $n\geq t$.
We say that an algorithm $\cal M$ solves $(X^d,n,t)$ with parameters $(\Delta,w)$ if for every input database $S\in (X^d)^n$, algorithm ${\cal M}(S)$ outputs a center $c$ and a radius $r$ s.t.\ the following holds with high probability:
\begin{enumerate}
	\item The ball of radius $r$ around $c$ contains at least $t-\Delta$ input points (from $S$).
	\item Let $r_{opt}$ be the radius of the smallest ball in $X^d$ containing at least $t$ input points. Then $r\leq w\cdot r_{opt}$.
\end{enumerate}
\end{definition}

The 1-cluster problem is very natural on its own, and furthermore, an algorithm for solving the 1-cluster problem can be used as a building block in other applications:\\

\noindent{\bf Data exploration.}\; The 1-cluster problem has direct implications to performing data exploration privately, and, specifically to clustering. For example, one can think of an application involving map searches where one is interested in privately locating areas of certain ``types'' or ``classes'' of a given population to gain some insight of their concentration over different geographical areas.\\

\noindent{\bf Outlier detection.}\; Consider using a solution to the 1-cluster problem to locate a small ball containing, say, 90\% of the input points. This can be used as a basic private identification of outliers: The outcome of the algorithm can be viewed as defining a predicate $h$ that evaluates to one inside the found ball and to zero otherwise. $h$ can hence be useful for screening the inputs to a private analysis of the set of outlier points in the data.
Outliers can skew and mislead the training of classification and regression algorithms, and hence, excluding them from further (privacy preserving) analysis can increase accuracy.

Furthermore, outlier detection can help in reducing the noise level required for the differentially private analysis itself, which in many cases would result in a dramatic improvement in accuracy. To see how this would happen, recall that the most basic construction of differentially private algorithms is via the framework of {\em global sensitivity}~\cite{DMNS06} (see also Section~\ref{sec:prelim} below). Noise is added to the outcome of a computation, and the noise magnitude is scaled to the {\em sensitivity} of the computation, i.e., the worst-case difference that a change of a single entry of the database may incur. Restricting the input space to a ball of (hopefully) a small diameter typically results in a smaller global sensitivity, and hence also significantly less noise.\\

\noindent{\bf Sample and aggregate.}\;
Maybe most importantly, an algorithm for the 1-cluster problem can be used in the Sample and Aggregate technique~\cite{NRS07}. This generic technique allows using ``off the shelf'', {\em non-privacy preserving}, analyses and transforms their outcome so as to preserve differential privacy. 

Consider a (non-private) analysis $f$ mapping databases to (a finite subset of) $\R^d$, and assume that $f$ can be well approximated by evaluating $f$ on a random subsample taken from the database. In the sample and aggregate framework, instead of applying the analysis $f$ on the entire dataset, it is applied on several (say $k$) random sub-samples of the input dataset, obtaining $k$ outputs $S=\{x_1,x_2,\dots,x_k\}$ in $\R^d$. The outputs are then {\em aggregated} to give a privacy-preserving result $z$ that is ``close'' to some of the points in $S$. If $f$ has the property that results that are in the vicinity of ``good'' results are also ``good'', then $z$ will also be a ``good'' result. Furthermore, it suffices that (only) the aggregation procedure would be differentially private to guarantee that the entire construction satisfies differential privacy.

Using their aggregation function (discussed below), Nissim et al.\ constructed differentially private algorithms for $k$-means clustering and for learning mixtures of Gaussians~\cite{NRS07}. Smith used the paradigm in dimension $d=1$ to construct private statistical estimators~\cite{Smith11}. One of the most appealing features of the paradigm is it that allows transforming programs that were not built with privacy in mind into differentially private analyses. For example, GUPT~\cite{MTSSC12} is an implementation of differential privacy that uses differentially private averaging for aggregation. The development of better aggregators enables making the sample and aggregate paradigm more effective.

\subsection{Existing techniques}
As we will show (by reduction to a lower bound of Bun et al.~\cite{BNSV15}), solving the 1-cluster problem on infinite domains is impossible under differential privacy (for reasonable choices of parameters), so any private solution must assume a finite universe $X^d\subseteq \R^d$. We will consider the case that $X^d$ is a discrete grid, identified with the real $d$ dimensional unit cube quantized with grid step $1/(|X|-1)$.

\begin{table*}[ht!]
\begin{small}
\begin{center}
\begin{tabular}{|c"c|c|c|}
\hline
   & \specialcell{Needed cluster size -- $t$\\ Additive loss in cluster size -- $\Delta$} & \specialcell{Approximation factor\\ in radius -- $w$} & Running time\\[2ex]
\thickhline
\specialcell{Private\\aggregation~\cite{NRS07}}  & \specialcell{\vspace{-7px}\\$t \geq \max\left\{0.51n,O(\frac{d^2}{\epsilon^2}\log^2|X|)\right\}$ \\ $\Delta=0$\;\;\;\;\;\;\;\;\;\;\;\;\;\;\;\;\;\;\;\;\;\;\;\;\;\;\;\;\;\;\;\;\;\;\;\;\;\;\;\;\;\;\;} & $w=O(\sqrt{d}/\epsilon)$ & $\poly(n,d,\log|X|)$ \\[3ex]
\hline
\specialcell{\vspace{-7px}\\Exponential\\ mechanism~\cite{MT07}} & $t\geq\Delta= \tilde{O}(d)\cdot\log^2(|X|)/\epsilon$ & $w=1$ & $\poly(n,|X^d|)$\\[2ex]
\hline
\specialcell{\vspace{-7px}\\Query release\\ for threshold\\ functions~\cite{BNS13b,BNSV15}\\ {\bf ({\boldmath$d=1$} only)}} & 
\specialcell{$t\geq\Delta=\frac{1}{\epsilon}\cdot2^{(1+o(1))\log^*|X|}\cdot \log(\frac{1}{\delta})$\\[1ex] (ignoring $\polylog(n)$ factors)} & $w=1$ & $\poly(n,\log|X|)$\\[4ex]
\hline
This work & 
\specialcell{\vspace{-7px}\\
$t\geq \frac{\sqrt{d}}{\epsilon}\log^{1.5}\left(\frac{1}{\delta}\right)\cdot 2^{O(\log^*(|X|d))}$\\[1ex]
$\Delta=\frac{1}{\epsilon}\log\left(\frac{1}{\delta}\right)\cdot2^{O(\log^*(|X|d))}$\;\;\;\;\;\;
} 
& $w=O(\sqrt{\log n})$ & $\poly(n,d,\log|X|)$ \\[4ex]
\hline
\end{tabular}
\caption{\small \label{fig:compare}  Comparing different solutions from past work and our result.}
\end{center}
\end{small}
\end{table*}

We now list a few existing techniques that can be used to solve the 1-cluster problem $(X^d,n,t)$:\\

\noindent{\bf Private aggregation.} 
Nissim, Raskhodnikova, and Smith~\cite{NRS07} introduced an efficient algorithm capable of identifying a ball of radius $O(r_{opt}\cdot\sqrt{d}/\epsilon)$ containing at least $t$ points, provided that $t\geq 0.51n\geq O(\frac{d^2}{\epsilon^2}\log^2|X|)$.\footnote{
The results of~\cite{NRS07} do not assume a finite discrete grid universe. Instead, they allow both a multiplicative and an additive error in the radius of the found ball. The additive error in the radius is eliminated whenever $n\geq O(\frac{d^2}{\epsilon^2}\log^2|X|)$. More specifically, let $z_0$ denote the center of the smallest ball containing $t>0.51n$ input points. The algorithm of~\cite{NRS07} computes a center $z$ s.t.\ that the error vector $(z_0-z)$ has magnitude $O(\frac{r_{opt}}{\epsilon})+\frac{1}{\epsilon}\cdot e^{-\Omega(\epsilon\sqrt{n}/d)}$ in each coordinate.} There are three downsides here: (1)~The error in the radius of the found ball grows with $\sqrt{d}$, which might be unacceptable in high dimensions. (2)~The database size $n$ needs to be as big as $d^2\log^2|X|$. (3)~The algorithm can only identify a majority size cluster. If, e.g., the input points are split between several small balls such that none of them contains a majority of the points, then the algorithm results in an uninformative center $z$ chosen almost at random.\\

\noindent{\bf Exponential mechanism.} One of the first ideas for solving the 1-cluster problem is to use the exponential mechanism of McSherry and Talwar~\cite{MT07} to choose among all balls: Given a radius $r$ s.t.\ there exists a ball of radius $r$ in $X^d$ containing $t$ points, the exponential mechanism is capable of identifying a ball of radius $r$ containing $t-O(\log(|X|^d)/\epsilon)$ points. Privately finding the radius $r$ could be done using a binary search, which would increase the loss in the size of the cluster by a factor of $O(\log(\sqrt{d}|X|))$. 
Overall, this strategy results in a ball of radius $r_{opt}$ containing $t-\tilde{O}(d)\cdot\log^2(|X|)/\epsilon$ input points. 
Thus, the exponential mechanism can identify clusters even when they do not contain a majority of the points.
However, we are seeking for an algorithm with running time $\poly(n,d,\log|X|)$, while the exponential mechanism runs in time $\poly(n,|X^d|)$.\\

\noindent{\bf Query release for threshold functions.}
For the special case where $d=1$, the 1-cluster problem can be solved using algorithms for ``query release for threshold functions'': On input a database $S\in X^n$, a query release mechanism for threshold functions privately computes a database $S'\in X$ such that for every interval $I\subseteq X$ it holds that the number of points in $S'$ differs from the number of points in $S$ that lie in $I$ by at most $\Delta$. Searching for a smallest interval in $S'$ containing $\gtrsim t$ points results in an interval of length $2r_{opt}$ (that is, of radius $r_{opt}$) containing at least $t-O(\Delta)$ input points.
Known algorithms for query release for threshold functions~\cite{BNS13b,BNSV15} achieve $\Delta=\max\left\{2^{(1+o(1))\log^*|X|}\cdot \log(\frac{1}{\delta}),\polylog(n)\right\}/\epsilon$. Note that the dependency on $|X|$ has improved substantially from the $\polylog|X|$ of the above methods. Bun et al.\cite{BNSV15} also showed that $\Delta$ must be at least $\Omega(\log^*|X|)$ and hence this problem is impossible to solve for infinite $X$ (we show a similar lower bound for the 1-cluster problem).

\subsection{Our contributions}

We present an algorithm for the 1-cluster problem that achieves (almost) the best of all the above. Namely, it (a)~handles a minority size cluster, of size only sublinear in $d$ (better than all the above) and sublogarithmic in $|X|$ (as with query release), and loses even less than that in the size of the cluster; and (b)~avoids paying $d^{\Omega(1)}$ factors in the error of the radius (instead paying $O(\sqrt{\log n})$). See Table~\ref{fig:compare} for a comparison with past work.\\  

{
\renewcommand{\thetheorem}{\ref{thm:main}}
\begin{theorem}[Informal]
There exists an efficient $(\epsilon,\delta)$-differentially private algorithm that, given a set $S$ of $n$ points in a discretized $d$-dimensional cube $X^d$, and a parameter $t$, outputs a ball of radius $O\left( \sqrt{\log n} \cdot r_{opt} \right)$ of size at least $t-\frac{1}{\epsilon}\log\left(\frac{1}{\delta}\right)\cdot2^{O(\log^*(|X|d))}$,
provided that
$$t\geq \frac{\sqrt{d}}{\epsilon}\log^{1.5}\left(\frac{1}{\delta}\right)\cdot 2^{O(\log^*(|X|d))}.$$
\end{theorem}
\addtocounter{theorem}{-1}
}

We note that the algorithm of~\cite{NRS07} works in general metric output spaces, whereas ours is restricted to $\R^d$. We leave open the question of extending our construction to more general settings.

\section{Preliminaries}\label{sec:prelim}
\vspace{5px}
\noindent{\bf Notations.}\; Throughout the paper, we use $X$ to denote a finite totally ordered data universe, and use $X^d$ for the corresponding $d$-dimensional domain. We will identify $X^d$ with the real $d$-dimensional unit cube, quantized with grid step $1/(|X|-1)$. 
Datasets are (ordered) collections of elements from some data universe $U$ (e.g., $U=X$ or $U=X^d$). 
Two datasets $S,S'\in U^n$ are called {\em neighboring} if they differ on at most one entry, i.e., $S'=(S_{-i},x'_i)$ for some $1\leq i \leq |S|$ and $x'_i\in U$.\\


We will construct algorithms that use several differentially private mechanisms as subroutines, and analyze the overall privacy using the following composition theorem:

\begin{theorem}[\cite{DKMMN06, DworkLei}]\label{thm:composition1}
A mechanism that permits $k$ adaptive interactions with $(\epsilon,\delta)$-differentially private mechanisms (and does not access the database otherwise) is $(k\epsilon, k\delta)$-differentially private.
\end{theorem}

\subsection{The framework of global sensitivity~\cite{DMNS06}}


\begin{definition}[$L_p$-Sensitivity]
A function $f$ mapping databases to $\R^d$ has {\em $L_p$-sensitivity $k$} if $\|f(S)-f(S')\|_p\leq k$ for all neighboring $S,S'$.
\end{definition}

The most basic constructions of differentially private algorithms are obtained by adding noise calibrated to the global sensitivity of the computation. We will use the Laplace mechanism of~\cite{DMNS06} to obtain noisy estimations to counting queries (e.g., how many points in $S$ have 0 on their first coordinate? Such a query has sensitivity 1 since changing one database element can change the count by at most 1).

\begin{theorem}[Laplace mechanism~\cite{DMNS06}]\label{thm:lap}
A random variable is distributed as $\Lap(\lambda)$ if its probability density function is $f(y)=\frac{1}{2\lambda}\exp(-\frac{|y|}{\lambda})$. 
Let $\epsilon>0$, and let $f:U^* \rightarrow \R^d$ be a function of $L_1$-sensitivity $k$.
The mechanism $\AAA$ that on input $D\in U^*$ 
adds independently generated noise with distribution $\Lap(\frac{k}{\epsilon})$ to each of the $d$ output terms of $f(D)$ preserves $(\epsilon,0)$-differential privacy.
\end{theorem}

We will also use the Gaussian mechanism to obtain (noisy) averages of vectors in $\R^d$. See Appendix~\ref{sec:NoisyAVG} for details.

\begin{theorem}[Gaussian Mechanism \cite{DKMMN06}]\label{thm:gauss}
Let $\epsilon,\delta\in(0,1)$, and let $f:U^* \rightarrow \R^d$ be a function of $L_2$-sensitivity $k$. Denote $\sigma\geq\frac{k}{\epsilon}\sqrt{2\ln(1.25/\delta)}$.
The mechanism $\AAA$ that on input $D\in U^*$ 
adds independently generated noise with distribution $\NNN(0,\sigma^2)$ to each of the $d$ output terms of $f(D)$ preserves $(\epsilon,\delta)$-differential privacy.
\end{theorem}

\subsection{Stability based techniques~\cite{DworkLei, Adist, BNS13b}}
Given a database $S\in U^*$, consider the task of choosing a ``good'' solution out of a possible set of solutions $F$, and assume that this ``goodness'' is quantified using a {\em quality function} $q:U^*\times F\rightarrow\N$ assigning ``scores'' to solutions from $F$ (w.r.t.\ the given database $S$).
One of the most useful constructions in differential privacy -- the exponential mechanism~\cite{MT07} -- shows that such scenarios are compatible with differential privacy, and that an approximately optimal solution $f\in F$ can be privately identified provided that $q$ has low-sensitivity and that $|S|\gtrsim\log|F|$.

By limiting our attention to cases where the number of possible solutions with ``high'' scores is limited, it is possible to relax the requirement that $|S|\gtrsim\log|F|$, using what has come to be known as {\em stability based techniques}. In this work we use stability based techniques for the following task: Given a dataset $S\in U^n$ and a partition $P$ of $U$, find a set $p\in P$ containing (approximately) maximum number of elements of $S$. This task can be privately solved using algorithms for query release for point functions.

\begin{theorem}[\cite{BNS13b, Vadhan2016}]\label{thm:sanPoints}
Fix $\epsilon,\delta$.
Let $U$ be a data universe, let $P$ be a partition of $U$, and let $S\in U^n$ be an input database. 
There exists an $(\epsilon,\delta)$-differentially private algorithm s.t.\ the following holds. 
Let $T$ denote the maximum number of input elements (from $S$) that are contained in a set in $P$.
If $T\geq\frac{2}{\epsilon}\log(\frac{4n}{\beta\delta})$, then with probability at least $(1-\beta)$ the algorithm returns a set $q\in P$ containing at least $T-\frac{4}{\epsilon}\log(\frac{2n}{\beta})$ elements from $S$.
\end{theorem}

\section{Our algorithms}
In this paper we explore the following problem under differential privacy:

\begin{definition}[The Problem of a Minimal Ball Enclosing $t$ Points]
Given a set of $n$ points in the Euclidean space $\R^d$ and an integer $t\leq n$,
the goal is to find a ball of minimal radius $r_{opt}$ enclosing at least $t$ input points.

\end{definition}

To enhance readability, we are using this section as an informal presentation of our results, giving most of the ideas behind our construction. We will also briefly discuss some intuitive ideas which fail to solve the task at hand, but are useful for the presentation. Any informalities made hereafter will be removed in the sections that follow.

We start by recalling known facts (without concern for privacy) about the problem of a minimal ball enclosing $t$ points:

\begin{enumerate}
	\item It is NP-hard to solve exactly~\cite{Shenmaier13}.
	\item Agarwal et al.~\cite{Agarwal05} presented an approximation scheme (PTAS) which computes a ball of radius $(1 + \alpha)r_{opt}$ containing $t$ points in time $O(n^{1/\alpha}d)$.
	\item There is a trivial algorithm for computing a ball of radius $2 r_{opt}$ containing $t$ points: Consider only balls centered around input points, and return the smallest ball containing $t$ points.
	
Indeed, let $B$ denote a ball of radius $r_{opt}$ enclosing at least $t$ input points, and observe that a ball of radius $2 r_{opt}$ around any point in $B$ contains all of $B$. Hence, there exists a ball of radius $2 r_{opt}$ around an input point containing at least $t$ points.
\end{enumerate}

We present a (roughly) $\sqrt{\log n}$-approximation algorithm satisfying differential privacy:

\begin{theorem}\label{thm:main}
Let $n,t,\beta,\epsilon,\delta$ be s.t.
$$t\geq O\left(\frac{\sqrt{d}}{\epsilon}\log\left(\frac{1}{\beta}\right)\log\left(\frac{nd}{\beta\delta}\right)\sqrt{\log\left(\frac{1}{\beta\delta}\right)}\cdot 9^{\log^*(2|X|\sqrt{d})}\right).$$
There exists a $\poly(n,d,1/\beta,\log|X|)$-time $(\epsilon,\delta)$-differentially private algorithm
that solves the 1-cluster problem $(X^d,n,t)$ with parameters $(\Delta,w)$ and error probability $\beta$, where
$w=O\left( \sqrt{\log n} \right)$
and
$$\Delta=O\left(\frac{1}{\epsilon}\log\left(\frac{n}{\delta}\right)\log\left(\frac{1}{\beta}\right)\cdot9^{\log^*(2|X|\sqrt{d})}\right).$$
\end{theorem}

In words, there exists an efficient $(\epsilon,\delta)$-differentially private algorithm that (ignoring logarithmic factors) is capable of identifying a ball of radius $\tilde{O}(r_{opt})$ containing $t-\tilde{O}(\frac{1}{\epsilon})$ points, provided that $t\geq\tilde{O}(\sqrt{d}/\epsilon)$.

\begin{remark}
For simplicity, in the above theorem we identified $X^d$ with the real $d$-dimensional unit cube, quantized with grid step $1/(|X|-1)$. Our results trivially extend to domains with grid step $\ell$ and axis length $L=\max{X}-\min{X}$ by replacing $|X|$ with $L/\ell$.
\end{remark}

\begin{remark}
Observe that the parameters $t,\Delta$ in Theorem~\ref{thm:main} have some dependency on the domain size $|X|$. Although this dependency is very weak, it implies that our construction cannot be applied to instances with infinite domains. In Section~\ref{sec:lower} we show that this is a barrier one cannot cross with differential privacy, and that privately solving the 1-cluster problem on infinite domains is impossible (for reasonable choices of parameters).
\end{remark}

\begin{observation}
Our construction could be used as a heuristic for solving a $k$-clustering-type problem: Letting $t=n/k$, we can iterate our algorithm $k$ times and find a collection of (at most) $k$ balls that cover most of the data points. Using composition to argue the overall privacy guarantees, we can have (roughly) $k\lesssim \frac{(\epsilon n)^{2/3}}{d^{1/3}}$. 
\end{observation}

Towards proving Theorem~\ref{thm:main} we design two algorithms. The first, \texttt{GoodRadius}, is given as input a collection $S$ of $n$ points and a parameter $t$, and returns a radius $r$ such that there exists a ball of radius $r$ containing $\gtrsim t$ of the points in $S$ and, furthermore, $r$ is within a constant factor of the smallest such ball. 

The second algorithm, \texttt{GoodCenter}, is given as input the set $S$ of input points, a parameter $t$, and a radius $r$ computed by \texttt{GoodRadius}. The algorithm outputs a center $z$ of a ball of radius $\tilde O(r)$ containing $\gtrsim t$ of the points in $S$.
So, a simplified overview of our construction is:
\begin{enumerate}[leftmargin=5pt, itemsep=1pt]
\item[] Input: A set $S$ of $n$ points in $X^d$, and an integer $t\leq n$.
	\item[] {\em Step 1:} Identify a radius $r=O(r_{opt})$ s.t.\ there is a ball of radius $r$ containing $\gtrsim t$ input points.
	\item[] {\em Step 2:} Given $r$, locate a ball of radius $O\left(\sqrt{\log n}\cdot r\right)$ containing $\gtrsim t$ input points.
\end{enumerate}

\subsection{Finding the cluster radius: algorithm GoodRadius}\label{sec:GoodRadiusWarmup}
Let $S=(x_1,\ldots,x_n)$ be a database containing $n$ points in $X^d$. 
Given $S$ and $t\leq n$, our current task is to approximate the minimal radius $r_{opt}$ for which there is a ball of that radius containing at least $t$ points from $S$. 

We start with the following notations:
For a radius $r\geq0$ and a point $p\in \R^d$, let $\BBB_r(p)$ denote the number of input points contained in a ball of radius $r$ around $p$. That is,
$\BBB_r(p) = |\{i : \|x_i-p\|_2\leq r\}|$.

Recall that we are looking for (a radius of) a ball containing $\gtrsim t$ points from $S$, and that a ball containing $t$ points is just as good as a ball of the same radius containing $100t$ points. Hence, we modify our notation of $\BBB_r(p)$ to cap counts at $t$:
$$
\bar{\BBB}_r(p) = \min\Big\{ \; \BBB_r(p) \; , \; t \; \Big\}.
$$

Using that notation, our goal is to approximate 
$r_{opt} = \min\left\{ r\geq0 \; : \; \exists p\in\R^d \text{ s.t.\ }  \bar{\BBB}_r(p)\geq t  \right\}$. Recall that a direct computation of $r_{opt}$ is NP-hard,
and let us turn to the simple 2-approximation algorithm that considers only balls centered at input points. To that end, for every $r\geq0$ define $L(r)$ as the maximum number of input points contained in a ball of radius $r$ around some input point (capped at $t$). That is,
$$
L(r)=\max_{x_i\in S}\left\{ \bar{\BBB}_r(x_i) \right\}.
$$
As we next explain, using that notation, it suffices to compute a radius $r$ s.t. 
$$(i)\;\; L(r) \gtrsim t \qquad\text{ and }\qquad (ii)\;\;L(r/2) < t.$$
We now argue that such an $r$ satisfies the requirements of \texttt{GoodRadius}. By (i) there exists a ball of radius $r$ containing $\gtrsim t$ input points, so we just need to argue that $r\leq O(r_{opt})$. Assume towards contradiction that $r_{opt}<r/4$, and hence there exists a subset $D\subseteq S$ of $t$ input points which are contained in a ball of radius $r/4$. Observe that a ball radius $r/2$ around any point in $D$ contains all of $D$, and therefore $L(r/2)\geq t$, contradicting (ii).

So we only need to compute a radius $r$ satisfying properties (i) and (ii) above. However, the function $L$ has high sensitivity, and hence it is not clear how to privately estimate $L(r)$ for a given radius $r$.
To see why $L$ has high sensitivity, consider a set $S$ consisting of the unit vector $\vec{e_1}$ along with $t/2$ copies of the zero vector and $t/2$ copies of the vector $2\cdot \vec{e_1}$. So, a ball of radius $1$ around $\vec{e_1}$ contains all of the points, and $L(1)=t$. However, if we were to switch the vector $\vec{e_1}$ to $2\cdot\vec{e_1}$, then the ball around $\vec{e_1}$ is no longer valid (since we only consider balls centered around input points), and every existing ball of radius $1$ contains at most $t/2$ point. So the sensitivity of the function $L$ is $\Omega(t)$.

In order to reduce the sensitivity of $L$ we now redefine it using averaging (a related idea was also used in~\cite{NRS07}). For $r\geq0$ redefine $L$ as
$$L(r)=\frac{1}{t} \max\limits_{\text{distinct}~i_1,\ldots,i_t\in [n]}\left\{ \bar{\BBB}_r(x_{i_1})+\ldots+\bar{\BBB}_r(x_{i_t})\right\}.$$
That is, to compute $L(r)$ we construct a ball of radius $r$ around every input point, count the number of points contained in every such ball (counts are capped at $t$), and compute the average of the $t$ biggest counts.

To see that the redefined function $L(r)$ has low sensitivity, consider a set of $n$ input points and a ball of radius $r$ around every input point. Adding a new input point can increase by at most 1 the number of points contained within every existing ball. In addition, we now have a new ball centered around the new input point, and as we cap counts at $t$, we count at most $t$ points in this ball. Overall, adding the new input point can increase $L(r)$ by at most $t\cdot \frac{1}{t}+\frac{t}{t}=2$. The function $L(r)$ has therefore sensitivity $O(1)$. 

Utility wise, we are still searching for an $r$ s.t.
$$(i)\;\; L(r) \gtrsim t \qquad\text{ and }\qquad (ii)\;\;L(r/2) < t.$$
Again, such an $r$ is useful since, by (i), there exists a ball of radius $r$ containing $\gtrsim t$ input points, and by (ii) we have that $r\leq4 r_{opt}$: Otherwise (if $r_{opt}<r/4$) there exists a subset $D\subseteq S$ of $t$ points which are contained in a ball of radius $r/4$. A ball of radius $r/2$ around {\em every} point in $D$ contains all of $D$ and therefore there are $t$ balls of radius $r/2$ containing $t$ points. Hence, $L(r/2)\geq t$, contradicting (ii).

So, the function $L$ has low sensitivity, and we are searching for an $r$ s.t.\ $L(r)\gtrsim t$ and $L(r/2)<t$. This can easily be done privately using binary search with noisy estimates of $L$ for the comparisons, but as there are (roughly) $\log(\sqrt{d}|X|)$ comparisons such a binary search would only yield a radius $r$ s.t. $L(r)\gtrsim t-\log(\sqrt{d}|X|)$.\footnote{Alternatively, an $r$ s.t.\ $L(r)\gtrsim t$ and $L(r/2)<t$  could be privately computed using the sparse vector technique, which also yields a radius $r$ s.t. $L(r)\gtrsim t-\log(\sqrt{d}|X|)$.} In Section~\ref{sec:GoodRadiusAnalysis} we will use a tool from~\cite{BNS13b} (recursion on binary search) to improve the guarantee to $L(r)\gtrsim t-9^{\log^*(\sqrt{d}|X|)}$.

\begin{lemma}[Algorithm \texttt{GoodRadius}]
Let $S\in(X^d)^n$ be a database containing $n$ points from $X^d$ and let $t,\beta,\epsilon,\delta$ be parameters.
There exists a $\poly(n,d,\log|X|)$-time $(\epsilon,\delta)$-differentially private algorithm that on input $S$ outputs a radius $r\in\R$ s.t.\ with probability at least $(1-\beta)$:
\begin{enumerate}
	\item There is a ball in $X^d$ of radius $r$ containing at least $t-O\left( \frac{1}{\epsilon} \log(\frac{1}{\beta\delta}) \cdot 9^{\log^*(|X|\cdot d)} \right)$ input points.
	\item Let $r_{\mbox{opt}}$ denote the radius of the smallest ball in $X^d$ containing at least $t$ points from $S$. Then $r\leq 4\cdot r_{\mbox{opt}}$.
\end{enumerate}
\end{lemma}

\subsection{Locating a cluster: algorithm GoodCenter}

Let $r$ be the outcome of Algorithm \texttt{GoodRadius} (so $r=O(r_{opt})$ and there exists a ball of radius $r$ containing $\gtrsim t$ input points).
Given the radius $r$, our next task is to locate, with differential privacy, a small ball in $\R^d$ containing $\gtrsim t$ input points. 
We begin by examining two intuitive (but unsuccessful) suggestions for achieving this goal.

\paragraph{First Attempt.} One of the first ideas for using the given radius $r$ in order to locate a small ball is the following: Divide each axis into intervals of length $\approx r$, identify (for every axis) a ``heavy'' interval containing lots of input points, and return the resulting axis-aligned box. Such a ``heavy'' interval could be privately identified (on every axis) using known stability-based techniques~\cite{DworkLei, Adist, BNS13b}.

The main problem with our first attempt is that the resulting box might be empty. This is illustrated in Figure~\ref{fig:attempt1}, where a ``heavy'' interval is identified on each axis s.t.\ their intersection is empty.

\begin{figure*}[t]
\begin{center}
\begin{tikzpicture}[xscale=1,yscale=1,show background rectangle,inner frame sep=10pt]
\draw [<->] (0,2.8) -- (0,0) -- (10,0);
\draw [|-|] (0,0) -- (0,0.75);
\draw [-|] (0,0.75) -- (0,1.5);
\draw [-|] (0,1.5) -- (0,2.25);
\draw [-,dashed] (0,1.5) -- (9.5,1.5);
\draw [-,dashed] (0,2.25) -- (9.5,2.25);
\draw [-,line width=0.5mm] (0,1.5) -- (0,2.25);

\draw [|-|] (0,0) -- (0.75,0);
\draw [-|] (0.75,0) -- (1.5,0);
\draw [-|] (1.5,0) -- (2.25,0);
\draw [-|] (2.25,0) -- (3,0);
\draw [-|] (3,0) -- (3.75,0);
\draw [-|] (3.75,0) -- (4.5,0);
\draw [-|] (4.5,0) -- (5.25,0);
\draw [-|] (5.25,0) -- (6,0);
\draw [-|] (6,0) -- (6.75,0);
\draw [-|] (6.75,0) -- (7.5,0);
\draw [-|] (7.5,0) -- (8.25,0);
\draw [-|] (8.25,0) -- (9,0);
\draw [-,dashed] (8.25,0) -- (8.25,2.5);
\draw [-,dashed] (9,0) -- (9,2.5);
\draw [-,line width=0.5mm] (8.25,0) -- (9,0);



\node[fill,circle,inner sep=0pt,minimum size=3pt] (a1) at (0.7,1.9) {};
\node[fill,circle,inner sep=0pt,minimum size=3pt] (a1) at (1,1.9) {};
\node[fill,circle,inner sep=0pt,minimum size=3pt] (a1) at (1.2,1.9) {};
\node[fill,circle,inner sep=0pt,minimum size=3pt] (a1) at (1.6,1.9) {};
\node[fill,circle,inner sep=0pt,minimum size=3pt] (a1) at (1.8,1.9) {};
\node[fill,circle,inner sep=0pt,minimum size=3pt] (a1) at (2,1.9) {};

\node[fill,circle,inner sep=0pt,minimum size=3pt] (a1) at (8.6,0.2) {};
\node[fill,circle,inner sep=0pt,minimum size=3pt] (a1) at (8.6,0.4) {};
\node[fill,circle,inner sep=0pt,minimum size=3pt] (a1) at (8.6,0.6) {};
\node[fill,circle,inner sep=0pt,minimum size=3pt] (a1) at (8.6,0.9) {};
\node[fill,circle,inner sep=0pt,minimum size=3pt] (a1) at (8.6,1.1) {};
\node[fill,circle,inner sep=0pt,minimum size=3pt] (a1) at (8.6,1.3) {};

\node[fill,circle,inner sep=0pt,minimum size=3pt] (a1) at (3.2,0.15) {};
\node[fill,circle,inner sep=0pt,minimum size=3pt] (a1) at (3.3,0.3) {};
\node[fill,circle,inner sep=0pt,minimum size=3pt] (a1) at (3.15,0.4) {};
\node[fill,circle,inner sep=0pt,minimum size=3pt] (a1) at (3.5,0.36) {};
\node[fill,circle,inner sep=0pt,minimum size=3pt] (a1) at (3.6,0.6) {};
\node[fill,circle,inner sep=0pt,minimum size=3pt] (a1) at (3.4,0.72) {};
\node[fill,circle,inner sep=0pt,minimum size=3pt] (a1) at (3.35,0.5) {};

\end{tikzpicture}
\end{center}
\vspace{-15px}\caption{\small An illustration of ``heavy'' intervals s.t.\ their intersection is empty. \label{fig:attempt1}}
\end{figure*}
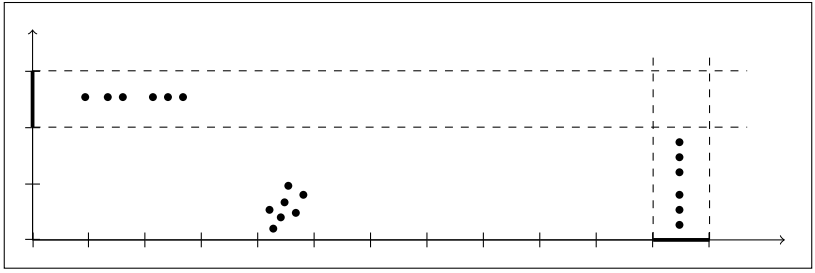

\paragraph{Second Attempt.} The failure point of our first strategy was the attempt to locate the cluster in an axis by axis manner. Trying to avoid that pitfall, consider the following idea for identifying a ``heavy'' {\em box} in $\R^d$:
Let us denote by $P\subseteq S$ the guaranteed set of $\gtrsim t$ input points which are contained in a ball of radius $r$ in $\R^d$.
Observe that the set $P$ is of diameter $2r$, and divide each axis into {\em randomly shifted} intervals of length $\approx 4 d r$. For every axis we have that the projection of $P$ onto that axis is contained within one interval w.p.\ $\gtrsim 1-1/(2d)$, and using the union bound, this is the case for all axes simultaneously w.p.\ $\gtrsim 1/2$. That is, without looking at the data, we have partitioned $\R^d$ into disjoint boxes of side length $\approx 4dr$ s.t.\ at least one of them contains $\gtrsim t$ input points, and such a ``heavy'' box can be privately identified using known stability-based techniques. While the resulting box is indeed ``heavy'', it is of side-length $\approx dr$ (i.e., of diameter $\approx d^{1.5} r$), which is not what we are looking for.

\paragraph{Towards a Solution.} Assume (for now) that we have privately identified a (concisely described) subset $X'$ of $X$ such that $S'=S\cap X'$ has $\gtrsim t$ points and is contained in a ball of radius $r$. Our current goal is, therefore, to identify a small ball enclosing all of $S'$. One option (which still does not preserve privacy, but has potential) it the following: Compute the average $c$ of the points in $S'$ and return a ball of radius $r$ around $c$. This option has merit since computing the average of input points can be made private by adding random noise to every coordinate, with magnitude proportional to the diameter of our subset $S'$ divided by its size $|S'|\gtrsim t$ (the intuition is that random noise of that magnitude masks any possible change limited to one input element, see Theorem~\ref{thm:gauss}). In our case, we would like to use $r$ (the diameter of $S'$) as such a bound, and hence obtain a ball of radius (roughly) $2r$. However, all of our discussion above only holds with high probability, say with probability $1-\beta$. In particular, the diameter of $S'$ is only bounded with probability $1-\beta$. In order for the privacy analysis to go through, we need this bound to hold with probability at least $1-\delta$, i.e., set $\beta=\delta$. Since $\delta$ is typically a negligible function of $n$, and since our running time depends on $1/\beta$, this is unacceptable.

As we next explain, our first (failed) attempt comes in handy for bounding the necessary noise magnitude. For the intuition, recall that our first attempt failed because we were misled by points outside the small cluster. By limiting our attention only to points in $S'$ (which are clustered), this is no longer an issue.

Assume that the set $S'$ contains $\gtrsim t$ points and that its diameter is $r$, and consider the following procedure: Partition every axis of $\R^d$ into intervals of length $r$. On every axis, at least one such interval contains (the projection of) $\gtrsim t/2$ points, and we can find such a ``heavy'' interval $I$ using known stability-based techniques. Afterwards, we can extend its length by $r$ to the left and to the right to obtain an interval $\hat{I}$ of length $3r$ containing all of $S'$. See Figure~\ref{fig:containing-int} for an illustration. So, on every axis we identified an interval of length $3r$ containing {\em all} of the points in $S'$. Hence, the intersection of all those intervals is a box $B$ of diameter $\approx \sqrt{d}r$ containing all of $S'$.

\begin{figure*}[t]
\begin{center}
\begin{tikzpicture}[xscale=1,yscale=1,show background rectangle,inner frame sep=10pt]
\draw [<->] (0,2) -- (0,0) -- (10,0);

\draw [|-|] (0,0) -- (0.75,0);
\draw [-|] (0.75,0) -- (1.5,0);
\draw [-|] (1.5,0) -- (2.25,0);
\draw [-|] (2.25,0) -- (3,0);
\draw [-|] (3,0) -- (3.75,0);
\draw [-|] (3.75,0) -- (4.5,0);
\draw [-|] (4.5,0) -- (5.25,0);
\draw [-|] (5.25,0) -- (6,0);
\draw [-|] (6,0) -- (6.75,0);
\draw [-|] (6.75,0) -- (7.5,0);
\draw [-|] (7.5,0) -- (8.25,0);
\draw [-|] (8.25,0) -- (9,0);
\draw [-,line width=0.5mm] (3.75,0) -- (4.5,0);


\node at (4.1,0.3) {$I$};
\node[xscale=7,yscale=1.25,rotate=90] at (4.115,-0.4) {$\{$};
\node at (4.1,-0.85) {$\hat{I}$};

\node[fill,circle,inner sep=0pt,minimum size=3pt] (a1) at (3.5,0.85) {};
\node[fill,circle,inner sep=0pt,minimum size=3pt] (a1) at (3.6,1) {};
\node[fill,circle,inner sep=0pt,minimum size=3pt] (a1) at (3.45,1.1) {};
\node[fill,circle,inner sep=0pt,minimum size=3pt] (a1) at (3.8,1.06) {};
\node[fill,circle,inner sep=0pt,minimum size=3pt] (a1) at (3.9,1.3) {};
\node[fill,circle,inner sep=0pt,minimum size=3pt] (a1) at (3.7,1.42) {};
\node[fill,circle,inner sep=0pt,minimum size=3pt] (a1) at (3.65,1.2) {};
\node[fill,circle,inner sep=0pt,minimum size=3pt] (a1) at (4.05,1.1) {};

\end{tikzpicture}
\end{center}
\vspace{-15px}
\caption{\small An illustration of an interval $I$ of containing some of the points of $S'$, and the corresponding interval $\hat{I}$ of length $3|I|$ containing all of $S'$ (since $S'$ is of diameter $r=|I|$). \label{fig:containing-int}}
\end{figure*}
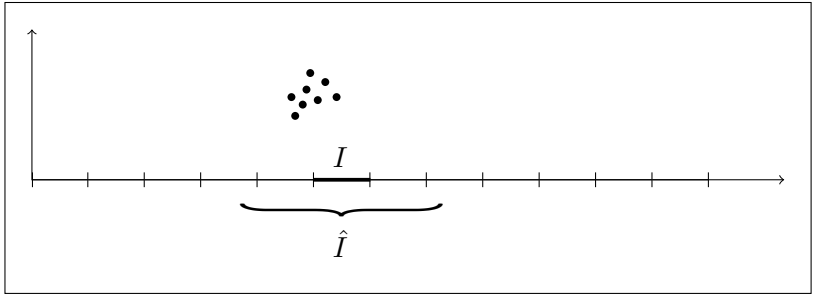

The thing that works in our favor here is that the above procedure {\em always} returns a box $B$ of diameter $\approx \sqrt{d}r$, even if our assumptions on the set $S'$ are invalid (in which case the box $B$ might be empty, but its diameter is the same). Now consider the set $\tilde{S}$ where we truncate all points in $S'$ to lie in $B$. Observe that (w.h.p.) we have that $S'\subseteq B$ and $\tilde{S}=S'$, and that, in any case, the diameter of $\tilde{S}$ is at most $\approx \sqrt{d}r$. We can therefore  privately release the noisy average of the points in $\tilde{S}$.
Assuming that $|S'|$ is big enough, the incurred noise is of magnitude $\lesssim r$, which results in a ball of radius $O(r)$ containing all of $S'$.

To summarize, it suffices to privately ``ignore'' all input points but $\gtrsim t$ points falling in some ball of radius $\approx r$.

\paragraph{Final Step.} Our final task is to identify a subset $S'\subseteq S$ of $\approx t$ input elements that are contained in a ball of radius roughly $r$.
Using the Johnson-Lindenstrauss transform we project our input points onto $\R^k$, where $k\approx\log(n)$ (w.h.p.\ point distances are preserved up to a constant factor). We denote the projection of a point $x\in \R^d$ as $f(x)\in \R^k$. 
By the properties of the JL-transform, it suffices to identify a part of the input $S'\subseteq S$ s.t.\ its projection
$f(S'):=\{f(x) : x\in S'  \}$ is contained within a ball of radius $\approx r$ in $\R^k$.

As we next explain, our second (unsuccessful) attempt could be used to identify such a subset $S'\subseteq S$. The intuition is that our second attempt incurred an unacceptable error factor of $\poly(d)$ in the cluster radius when locating the ball in $\R^d$, and this error factor is mitigated by locating the ball in the lower-dimensional space $\R^k$ (where $k=O(\log n)$).

As above, let $P\subseteq S$ be the guaranteed set of $\gtrsim t$ input points contained within a ball of radius $r$ in $\R^d$.
Note that (w.h.p.) the set $f(P):=\{f(x) : x\in P \}$ is contained within a ball of radius $\lesssim r$ in $\R^k$, and assume that this is the case. Partition every axis $i$ of $\R^k$ into randomly shifted intervals ${\cal I}_i=\{I^i_j\}_{j\in\Z}$ of length $\approx 4r$.
On every axis $i$, with probability $\gtrsim 1/2$ the projection of $f(P)$ on the $i^\text{th}$ axis is completely contained within one interval in ${\cal I}_i$. With probability $\gtrsim 0.5^k=1/\poly(n)$, this is the case for all of the $k$ axes simultaneously, and $f(P)$ is completely contained within an axis aligned box whose projection onto every axis $i$ of $\R^k$ is in ${\cal I}_i$. In other words, we have partitioned $\R^k$ into disjoint $k$-dimensional axis aligned boxes of side-length $\approx 4r$ s.t.\ with noticeable probability (we will later use repetitions to amplify this probability, and use the sparse vector technique to privately choose one of the repetitions) at least one of them contains $\gtrsim t$ (projected) input points. Such a ``heavy'' rectangle $B$ could be privately identified using stability based techniques (its diameter is $\approx r\sqrt{k}\approx r\sqrt{\log n}$). Finally, we define the set $S'=\{x\in S : f(x)\in B\}$ as the set of points that are mapped (by the JL transform) into the rectangle $B$. We now have that $S'$ contains $\gtrsim t$ input elements, since the box $B$ is ``heavy'' in $\R^k$, and the diameter $S'$ is $\lesssim r\sqrt{\log n}$, since that is the diameter of $B=f(S')$.
The complete construction appears in Algorithm \texttt{GoodCenter} (algorithm~\ref{alg:GoodCenter}).

\begin{lemma}[Algorithm \texttt{GoodCenter}]
Let $S\in(\R^d)^n$ be a database containing $n$ points in $\R^d$, and let $r,t,\beta,\epsilon,\delta$ be parameters s.t.\
$t\geq O\left(\frac{\sqrt{d}}{\epsilon}\log(\frac{1}{\beta})\log(\frac{nd}{\beta\epsilon\delta})\sqrt{\log(\frac{1}{\beta\delta})}\right)$.
There exists a $\poly(n,d,1/\beta)$-time $(\epsilon,\delta)$-differentially private algorithm that on input $S,r,t$ outputs a point $z\in \R^d$ s.t.\ the following holds.
If there exists a ball of radius $r$ in $\R^d$ containing at least $t$ points from $S$,
then with probability at least $1-\beta$, the ball of radius $O\left( r \sqrt{\log n}\right)$ around $z$ contains at least $t-O\left(\frac{1}{\epsilon}\log(\frac{1}{\beta})\log(\frac{n}{\beta\epsilon\delta})\right)$ of the points in $S$.
\end{lemma}


\section{Details of analysis}\label{sec:details}
\subsection{Algorithm GoodRadius}\label{sec:GoodRadiusAnalysis}

As we explained in Section~\ref{sec:GoodRadiusWarmup}, it is possible to compute an approximation for the optimal radius using a binary search on a carefully chosen low sensitivity function.
We use the following tool from~\cite{BNS13b} in order to reduce the sample cost of that binary search.

\begin{definition}
A function $Q(\cdot)$ over a totally ordered domain is {\em quasi-concave} if for every $i\leq \ell \leq j$ we have $Q(\ell)\geq \min\{Q(i),Q(j)\}$.
\end{definition}

\begin{definition}[\cite{BNS13b}]
A {\em Quasi-Concave Promise Problem} consists of an ordered set $F$ of possible solutions, a database $S\in U^n$, a sensitivity-1 quality function $Q:U^n\times F \rightarrow\R$, an approximation parameter $\alpha$, and another parameter $p$ (called a {\em quality promise}). 

If $Q(S,\cdot)$ is quasi-concave and if there exists a solution $f\in F$ for which $Q(S,f)\geq p$ then a good output for the problem is a solution $g\in F$ satisfying $Q(S,g)\geq(1-\alpha)p$. The outcome is not restricted otherwise.
\end{definition}

We will use Algorithm \texttt{RecConcave} from~\cite{BNS13b} to solve quasi-concave promise problems while preserving differential privacy:

\begin{theorem}[Algorithm \texttt{RecConcave}~\cite{BNS13b}]\label{thm:rec}
Let $U$ be a domain, let $F$ be a totally ordered (finite) set of solutions, and let $Q:U^n\times F \rightarrow\R$ be a sensitivity-1 quality function.
Let $\alpha,\beta,\epsilon,\delta$ be parameters. There exists an $(\epsilon,\delta)$-differentially private algorithm s.t.\ the following holds. On input a database $S$ and a quality promise $p$ for which $Q(S,\cdot)$ is quasi-concave and 
$$\max_{f\in F}\{Q(S,f)\}\geq p \geq 8^{\log^*|F|} \cdot \frac{36 \log^*|F|}{\alpha\epsilon} \log\Big(\frac{12\log^*|F|}{\beta\delta}\Big),$$
the algorithm outputs a solution $f\in F$ s.t.\ $Q(S,f)\geq(1-\alpha)p$ with probability at least $(1-\beta)$.
\end{theorem}

\begin{remark}
The computational efficiency of algorithm \texttt{RecConcave} depends on the quality function $Q$. It can be made efficient in cases where for every database $S\in U^n$, the totally ordered set of solutions $F$ can be partitions into $k=\poly(n)$ intervals of sequential solutions $F_1,F_2,\dots,F_k$ s.t.\ for every $i$ and for every $f,f'\in F_i$ we have $Q(S,f)=Q(S,f')$. In such cases, the algorithm runs in time $\poly(n,\log|F|)$, assuming that the partition of $F$, and that evaluating $Q$, can be done in time $\poly(n,\log|F|)$.
\end{remark}

\begin{algorithm*}[t]
\caption{\texttt{GoodRadius}}\label{alg:GoodRadius}

\begin{enumerate}[rightmargin=10pt,itemsep=1pt]

\item[] {\bf Input:} Database $S \in (X^d)^n$, desired ball volume $t$, failure probability bound $\beta$, and privacy parameters $\epsilon,\delta$.

\item[] {\bf Algorithm used:} Algorithm \texttt{RecConcave} for privately solving quasi-concave problems. We denote the minimal quality promise needed for algorithm \texttt{RecConcave} (for our choice of parameters) as $\Gamma=8^{\log^*(2|X|\sqrt{d})} \cdot \frac{144 \log^*(2|X|\sqrt{d})}{\epsilon} \log\Big(\frac{24\log^*(2|X|\sqrt{d})}{\beta\delta}\Big)$.

\item[] {\bf Notation:} For $x\in X^d$ and $0\leq r\in\R$ let $\BBB_r(x,S)$ denote the number of input points contained in a ball of radius $r$ around $x$. For $r<0$, let $\BBB_r(x,S)=0$. Let $\bar{\BBB}_r(x,S)=\min\{\BBB_r(x,S),t\}$.

\item For $r\in\R$ define $L(r,S)=\frac{1}{t} \max\limits_{\text{distinct}~i_1,\ldots,i_t\in [n]}\left( \bar{\BBB}_r(x_{i_1},S)+\ldots+\bar{\BBB}_r(x_{i_t},S)\right).$

\begin{enumerate}[label=\gray{\%},topsep=-10pt]
\item \gray{
That is, for every input point $x\in S$ we count the number of input points contained in a ball of radius $r$ around $x$, capped at $t$. We define $L(r,S)$ as the average of the $t$ largest counts.}
\item \gray{
Note that $L(\cdot,S)$ is a non-decreasing function.} 
\end{enumerate}

\item\label{step:zerocluster} Let $\tilde{L}(0,S)=L(0,S)+\Lap(4/\epsilon)$. If $\tilde{L}(0,S)>t-2\Gamma-\frac{4}{\epsilon}\ln(2/\beta)$, then halt and return $z=0$.
\begin{enumerate}[label=\gray{\%},topsep=-10pt]
\item \gray{
Step~\ref{step:zerocluster} handles the case where there exists a cluster of radius zero containing $\gtrsim t$ of the input points.}
\end{enumerate}

\item Define the quality function $Q(r,S)=\frac{1}{2}\min\left\{ t-L\left(r/2,S\right), \; L(r,S)-t+4\Gamma \right\}.$

\item\label{step:applyconcave} Apply algorithm \texttt{RecConcave} with privacy parameters $(\frac{\epsilon}{2},\delta)$, utility parameters  $(\alpha{=}\frac{1}{2},\frac{\beta}{2})$, quality function $Q$, and quality promise $\Gamma$ to choose and return $z\in\left\{0,\frac{1}{2|X|},\frac{2}{2|X|},\frac{3}{2|X|},\ldots,\left\lceil \sqrt{d} \right\rceil\right\}$.
\begin{enumerate} [label=\gray{\%},topsep=-10pt]
\item \gray{
For simplicity, we identify $X^d$ with the real $d$-dimensional unit cube, quantized with grid step $1/(|X|-1)$. Our results trivially extend to domains with grids steps $\ell$ by choosing the output out of $\left\{0,\frac{\ell}{2},\frac{2\ell}{2},\frac{3\ell}{2},\ldots,\left\lceil |X|\ell\sqrt{d} \right\rceil\right\}$.}
\end{enumerate}
\end{enumerate}

\end{algorithm*}

We now proceed with the privacy analysis of algorithm \texttt{GoodRadius}.

\begin{lemma}\label{lem:GoodRadiusPrivacy} 
Algorithm \texttt{GoodRadius} preserves $(\epsilon,\delta)$-differential privacy. 
\end{lemma}

\begin{proof}
Algorithm \texttt{GoodRadius} interacts with its input database in step~\ref{step:zerocluster} using the Laplace mechanism and in step~\ref{step:applyconcave} using algorithm \texttt{RecConcave}.
In order to show that those two interactions preserve privacy, we will now argue that $L(r,\cdot)$ is of sensitivity 2 (for every fixed $r$). To see why this is intuitively correct, consider a set of $n$ input points and a ball of radius $r$ around every input point. Adding a new input point can increase by at most 1 the number of points contained within every such ball. In addition, we now have a new ball centered around the new input point, and as we cap counts at $t$, we count at most $t$ points in this ball. Overall, adding the new input point can increase $L(r,\cdot)$ by at most $t\cdot \frac{1}{t}+\frac{t}{t}=2$.

More formally, let $S,S'$ be two neighboring databases and assume that $S'=S\setminus\{y\}\cup\{y'\}$ and that $S=S'\setminus\{y'\}\cup\{y\}$.
We have that
\begin{eqnarray*}
L(r,S) &=& \max\limits_{\substack{\text{distinct}\\x_1,\ldots,x_t\in S}}\frac{1}{t}\sum_{i=1}^t \bar{\BBB}_r(x_i,S)\\
&\leq&\hspace{-20px} \max\limits_{\substack{\text{distinct}\\ x_1,\ldots,x_{t-1}\in S\setminus\{y\}  }}\frac{1}{t}\left(t+\sum_{i=1}^{t-1} \bar{\BBB}_r(x_i,S)\right)\\
&=&\hspace{-20px} \max\limits_{\substack{\text{distinct}\\ x_1,\ldots,x_{t-1}\in S'\setminus\{y'\}  }}\frac{1}{t}\left(t+\sum_{i=1}^{t-1} \bar{\BBB}_r(x_i,S)\right)\\
&\leq&\hspace{-20px} \max\limits_{\substack{\text{distinct}\\ x_1,\ldots,x_{t-1}\in S'\setminus\{y'\}  }}\frac{1}{t}\left(t+\sum_{i=1}^{t-1} [\bar{\BBB}_r(x_i,S')+1]\right)\\
&\leq&\hspace{-10px} \max\limits_{\substack{\text{distinct}\\ x_1,\ldots,x_{t-1}\in S'  }}\frac{1}{t}\left(t+\sum_{i=1}^{t-1} [\bar{\BBB}_r(x_i,S')+1]\right)\\
&\leq& L(r,S')+2.
\end{eqnarray*}

Similarly, $L(S,r)\geq L(S',r)-2$, and $L(r,\cdot)$ is of sensitivity 2. Hence, the use of the laplace mechanism on step~\ref{step:zerocluster} preserves $(\frac{\epsilon}{2},0)$-differential privacy. Moreover, $\frac{1}{2}L(r,\cdot)$ is of sensitivity 1, and, therefore, $Q(r,\cdot)$ is of sensitivity 1 (defined as the minimum of two sensitivity 1 expressions). The application of algorithm \texttt{RecConcave} preserves $(\frac{\epsilon}{2},\delta)$-differential privacy. Overall, algorithm \texttt{GoodRadius} is $(\epsilon,\delta)$-differentially private by Composition Theorem~\ref{thm:composition1}.
\end{proof}

We now turn to proving the correctness of algorithm \texttt{GoodRadius}.

\begin{lemma}\label{lem:GoodRadiusUtility}
Let \texttt{GoodRadius} be executed on a database $S$ containing $n$ points in $X^d$ and on parameters $t,\beta,\epsilon,\delta$, and let $\Gamma$ be as defined in algorithm \texttt{GoodRadius}.
With probability at least $(1-\beta)$, the output $z$ satisfies: (1)~There exists a ball in $X^d$ of radius $z$ containing at least $(t-4\Gamma-\frac{4}{\epsilon}\ln(1/\beta))$ input points from $S$.
(2)~Let $r_{\mbox{opt}}$ denote the radius of the smallest ball in $X^d$ containing at least $t$ input points from $S$. Then $z\leq 4 r_{\mbox{opt}}$.
\end{lemma}

\begin{proof}
Note that if $L(0,S)\geq t-2\Gamma$, then \texttt{GoodRadius} fails to output $z=0$ in step~\ref{step:zerocluster} with probability at most $\beta/2$. We continue the proof assuming that $L(0,S)< t-2\Gamma$.

We now argue that algorithm \texttt{RecConcave} returns (w.h.p.) a value $z$ s.t. $Q(z,S)$ is significant. We need to show that $Q(\cdot,S)$ is quasi-concave, and that there exists an $r$ s.t.\ $Q(r,S)\geq\Gamma$. To see that $Q(\cdot,S)$ is quasi-concave, note that $L(\cdot,S)$ is non-decreasing (as a function of $r$), and that, hence, for every $r_1<r_2<r_3$
\begin{eqnarray*}
Q(r_2,S)&=&\min\left\{ \frac{t-L(r_2/2,S)}{2}, \frac{L(r_2,S)-t+4\Gamma}{2} \right\}\\
&\geq&\min\left\{ \frac{t-L(r_3/2,S)}{2}, \frac{L(r_1,S)-t+4\Gamma}{2}\right\}\\
&\geq& \min\left\{ Q(r_3,S), Q(r_1,S) \right\}.
\end{eqnarray*}

To see that there exists an $r$ s.t. $Q(r,S)\geq\Gamma$,
recall that $L(0,S)=L\left(\frac{1}{2|X|},S\right)<t-2\Gamma$, and note that $L(\sqrt{d},S)=n\geq t$. Now consider the smallest
$r\in\left\{0,\frac{1}{2|X|},\frac{2}{2|X|},\frac{3}{2|X|},\ldots,\left\lceil \sqrt{d} \right\rceil\right\}$
s.t.\ $L(r,S)\geq t-2\Gamma$. For that $r$ it holds that $$Q(r,S)=\frac{1}{2}\min\{ t-L(r/2,S), \; L(r,S)-t+4\Gamma \}\geq \frac{1}{2}\min\{ 2\Gamma, \; 2\Gamma \}=\Gamma.$$

By the properties of algorithm \texttt{RecConcave}, with probability at least $(1-\beta/2)$ the output $z$ is s.t.\ $Q(z,S)\geq\frac{\Gamma}{2}$. Hence, by the definition of $Q$ we have that 
$$\text{(a)}~L(z,S)\geq t-4\Gamma\quad\text{and}\quad\text{(b)}~L(z/2,S)\leq t-\frac{\Gamma}{2}.$$
Recall that $L(z,S)$ averages $\bar B_z(x,S)$ over $t$ points $x\in S$. Hence, by (a), there exists a ball of radius $z$ in $X^d$ that contains at least $t-4\Gamma$ input points from $S$.

Let $P\subseteq S$ be a set of $t$ input points, and assume towards contradiction that there is a ball of radius $z/4$ in $X^d$ that contains all of the points in $P$. Now note that a ball of radius $z/2$ around every point in $P$ contains all of the points in $P$. Hence, $L(z/2,S)\geq t$. This contradicts (b).
We conclude that Algorithm \texttt{GoodRadius} returns a good radius with probability at least $1-\beta$.
\end{proof}

\subsection{Additional preliminaries}

Before formally presenting algorithm \texttt{GoodCenter}, we introduce several additional tools.

\subsubsection{Composition theorems}
Recall that the privacy guaranties in composition theorem~\ref{thm:composition1}
deteriorates linearly with the number of interactions. By bounding the \emph{expected} privacy loss in each interaction (as opposed to worst-case), Dwork et al.~\cite{DRV10} showed the following stronger composition theorem, where privacy deteriorates (roughly) as $\sqrt{k}\epsilon+k\epsilon^2$ (rather than $k\epsilon$).

\begin{theorem}[\cite{DRV10}]\label{thm:composition2}
Let $\eps,\delta,\delta'>0$. A mechanism that permits $k$ adaptive interactions with $(\epsilon,\delta)$-differentially private mechanisms (and does not access the database otherwise) is $(\epsilon', k\delta+\delta')$-differentially private, for $\epsilon'=2k\epsilon^2+\epsilon\sqrt{2k\ln(1/\delta')}$.
\end{theorem}

\subsubsection{The sparse vector technique~\cite{DNRRV09}}\label{sec:sparseVector}

Consider a large number of low sensitivity functions $f_1,$ $f_2,\ldots, f_k$, which are given (one by one) to a data curator (holding a database $S$). Given a dataset $S$, Algorithm \texttt{AboveThreshold} by Dwork et al.~\cite{DNRRV09} identifies the queries $f_i$ whose value $f_i(S)$ is greater than some threshold $t$:

\begin{theorem}[Algorithm \texttt{AboveThreshold}]\label{thm:AboveThreshold}
There exists an $(\epsilon,0)$-differentially private algorithm $\cal A$ such that for $k$ rounds, after receiving a sensitivity-1 query $f_i:U^*\rightarrow\R$, algorithm $\cal A$ either outputs $\top$ and halts, or outputs $\bot$ and waits for the next round.
If $\cal A$ was executed with a database $S\in U^*$ and a threshold parameter $t$, then the following holds with probability $(1-\beta)$:
(i) If a query $f_i$ was answered by $\top$ then $f_i(S)\geq t-\frac{8}{\epsilon}\log(2k/\beta)$;
(ii) If a query $f_i$ was answered by $\bot$ then $f_i(S)\leq t+\frac{8}{\epsilon}\log(2k/\beta)$.
\end{theorem}

\subsubsection{Geometric tools}
We will use the following technical lemma to argue that if a set of points $P$ is contained within a ball of radius $r$ in $\R^d$, then by randomly rotating the Euclidean space we get that (w.h.p.) $P$ is contained within an axis-aligned rectangle with side-length $\approx r/\sqrt{d}$. 

\begin{lemma}[e.g.,~\cite{VaziraniRao}]\label{lem:RandomRotation}
Let $P \in (\R^d)^m$ be a set of $m$ points in the $d$ dimensional Euclidean space, and let $Z=(z_1,\ldots,z_d)$ be a random orthonormal basis for $\R^d$. Then,
$$\Pr_Z \left[\forall x,y \in P:\; \forall 1\leq i\leq d: \; \left| \langle x-y,  z_i\rangle\right| \leq 2\sqrt{\ln(dm/\beta)/d}\cdot \|x-y\|_2\right]\geq1-\beta.$$
\end{lemma}

We will use the Johnson Lindenstrauss transform to embed a set of points $S\in \R^d$ in $\R^k, k\ll d$ while preserving point distances.

\begin{lemma}[JL transform~\cite{JL84}]\label{thm:JL}
Let $S\subseteq\R^d$ be a set of $n$ points, and let $\eta\in(0,1/2)$. Let $A$ be a $k\times d$ matrix whose entries are iid samples from $\mathcal{N}(0,1)$, and define $f:\R^d\rightarrow\R^k$ as $f(x)=\frac{1}{\sqrt{k}}Ax$. Then,
$$
\Pr_A \left[
\begin{array}{l}
\forall x,y \in S \text{ it holds that: }\\
(1-\eta)\|x-y\|_2^2 \leq \|f(x)-f(y)\|_2^2 \leq (1+\eta)\|x-y\|_2^2
\end{array}
\right]\geq 1 - 2n^2\exp\left(-\frac{\eta^2 k}{8}\right).
$$
\end{lemma}

\subsection{Algorithm GoodCenter}\label{sec:GoodCenterAnalysis}

Given the outcome of Algorithm \texttt{GoodRadius}, we now show that algorithm \texttt{GoodCenter} privately locates a small ball containing $\gtrsim t$ points. We start with its privacy analysis.

\begin{lemma}\label{lem:GoodCenterPrivacy}
Algorithm \texttt{GoodCenter} preserves $(\epsilon,\delta)$-differential privacy. 
\end{lemma}

\begin{proof}
Algorithm \texttt{GoodCenter} interacts with its input database on steps~\ref{step:abovethreshold},~\ref{step:abovethreshold_1},~\ref{step:choosing},~\ref{step:choosing_1},~\ref{step:output_z}.
Steps~\ref{step:abovethreshold},~\ref{step:abovethreshold_1} initialize and use Algorithm \texttt{AboveThreshold}, which is $(\frac{\epsilon}{4},0)$-differentially private.
Step~\ref{step:choosing} invokes the algorithm from Theorem~\ref{thm:sanPoints} (to choose a ``heavy'' box $B$), which is $(\frac{\epsilon}{4},\frac{\delta}{4})$-private. 
Step~\ref{step:choosing_1} makes $d$ applications of the algorithm from Theorem~\ref{thm:sanPoints}. By theorem~\ref{thm:composition2} (composition), this preserves $(\frac{\epsilon}{4},\frac{\delta}{4})$-differential privacy.
Step~\ref{step:output_z} invokes the Gaussian mechanism, which is $(\frac{\epsilon}{4},\frac{\delta}{4})$-private. 
Overall, \texttt{GoodCenter} is $(\epsilon,\delta)$-differentially private by composition.
\end{proof}

We now proceed with the utility analysis of algorithm \texttt{GoodCenter}.

\begin{lemma}\label{lem:GoodCenterUtility}
Let \texttt{GoodCenter} be executed on a database $S$ containing $n$ points in $\R^d$ with $r,t,\beta,\epsilon,\delta$ s.t.
$$t\geq O\left(\frac{\sqrt{d}}{\epsilon}\log\left(\frac{nd}{\beta\delta}\right)\sqrt{\log(\frac{1}{\delta})}\right).$$
If there exists a ball of radius $r$ in $\R^d$ containing at least $t$ points from $S$,
then with probability at least $1-\beta$, the output $\hat{y}$ in Step~\ref{step:output_z} is s.t.\ at least $t-O\left(\frac{1}{\epsilon}\log(\frac{n}{\beta})\right)$ of the input points are contained in a ball of radius $O\left( r \sqrt{\log(\frac{n}{\beta})}\right)$ around $\hat{y}$.
\end{lemma}

\begin{remark}
The dependency in $1/\beta$ can easily be removed from the radius of the resulting ball by applying \texttt{GoodCenter} with a constant $\beta$ and amplifying the success probability using repetitions.
\end{remark}

\begin{proof}
First note that by Theorem~\ref{thm:JL} (the JL transform), with probability at least $1-\beta$, for every $x,y\in S$ it holds that $\|x-y\|$ and $\|f(x)-f(y)\|$ are similar up to a multiplicative factor of $(1\pm\frac{1}{2})$. We continue the proof assuming that this is the case. Hence, there exists a ball of radius $3r$ in $\R^k$ containing at least $t$ points from $\{f(x) : x\in S\}$. Denote this set of (at least $t$) projected points as $W$.

Clearly, the projection of the set $W$ onto any axis of $\R^k$ lies in an interval of length $3r$. Recall that on Step~3a we partition every axis into randomly shifted intervals $\{A_j^i\}$ of length $300r$. Hence, for every axis $i$ with probability $0.99$ it holds that the projection of $W$ onto the $i^\text{th}$ axis is contained within one of the $\{A_j^i\}$'s.
The probability that this holds simultaneously for all of the $k$ axes is $0.99^k\geq\frac{\beta}{2n}$.
Note that in such a case there exists a rectangle in $\{B_{\vec{j}}\}$ containing all of the points in $W$, and hence, the corresponding query $q$ (defined on step~5) satisfies $q(S)\geq t$. So, every (randomly constructed) query $q$ satisfies $q(S)\geq t$ with probability $\beta/(2n)$.

Although the iteration of steps~\ref{step:itr_begin}--\ref{step:itr_end} might be repeated less than $2n\log(1/\beta)/\beta$ times, imagine that all of the (potential) $2n\log(1/\beta)/\beta$ queries were prepared ahead of time (and some may have never issued to \texttt{AboveThreshold}). With probability at least $(1-\beta)$ at least one such query $q$ satisfies $q(S)\geq t$. We continue with the proof assuming that this is the case. Thus, by the properties of algorithm \texttt{AboveThreshold}, with probability at least $(1-\beta)$, the loop on Step~\ref{step:itr_end} ended with \texttt{AboveThreshold} returning $\top$. Moreover, in that iteration we have that $q(S)\geq t-\frac{200}{\epsilon}\log(2n/\beta)$. Thus, by the definition of $q$, after Step~\ref{step:itr_end} there exists a rectangle in $\{B_{\vec{j}}\}$ containing at least $t-\frac{200}{\epsilon}\log(2n/\beta)$ projected input elements.

By Theorem~\ref{thm:sanPoints}, with probability at least $(1-\beta)$, the box $B$ chosen on step~\ref{step:choosing} contains at least $t-\frac{216}{\epsilon}\log(\frac{2n}{\beta})$ projected input elements.
We continue the proof assuming that this is the case, and denote the set of input points from $S$ that are mapped into $B$ as $D$.

Note that $B$ is a box of diameter $300 r \sqrt{k}$. Hence, by our assumption on the projection $f$, for any $x,y\in D$ it holds that $\|x-y\|\leq 450 r \sqrt{k}$.
On Step~\ref{step:rotation} we generate a random basis $Z$ of $\R^d$. By Lemma~\ref{lem:RandomRotation}, with probability at least $1-\beta$, for every $x,y\in S$ and for every $z\in Z$ it holds that the projection of $(x-y)$ onto $z$ is of length at most $2\sqrt{\ln(\frac{d n}{\beta})/d}\cdot\|x-y\|$. Assuming that this is the case, the projection of $D$ onto any axis $z\in Z$ lies in an interval of length $p=900 r \sqrt{k \ln(\frac{d n}{\beta}) / d}$.
Therefore, when partitioning every axis $z_i\in Z$ into intervals $\III_i$ of length $p$ (on step~\ref{step:partitionP}), at least one interval $I\in\III_i$ contains at least half of the points in $D$. 
Assuming that $|D|\geq\frac{40\sqrt{d}}{\epsilon}\log(\frac{6nd}{\beta\delta})\sqrt{\ln(\frac{8}{\delta})}$,
for every axis $i$ Theorem~\ref{thm:sanPoints} ensures that with probability at least $(1-\beta/d)$ the chosen interval $I_i\in\III_i$ (on step~\ref{step:choosing_1}) contains at least one input elements from $D$. Recall that the projection of $D$ onto any axis $z_i$ lies in an interval of length $p$. Hence, letting $\hat{I}_i$ be $I_i$ after extending it by $p$ on each side, we ensure that $\hat{I}_i$ contains {\em all} of the points in $D$. So, the box in $\R^d$ whose projection onto every axis $z_i$ is $\hat{I}_i$ contains all of the points in $D$. Recall that (on step~\ref{step:center_c}) we defined $C$ to be the bounding sphere of that box, and hence, $C$ contains all of $D$, and $D'=D\cap C=D$.

Let $y$ denote the average of the point in $D'=D$, and observe that a ball of radius $450r\sqrt{k}$ around $y$ contains all of the points in $D$.
The output on step~\ref{step:output_z} is computed using the Gaussian mechanism as the noisy average of the points in $D'=D$, where the noise magnitude in every coordinate is proportional to the diameter of $C$ (which is $5400r\sqrt{k\cdot\ln(\frac{dn}{\beta})}$). By the properties of the Gaussian mechanism (see appendix~\ref{sec:NoisyAVG}), with probability at least $(1-\beta)$ the output $\hat{y}$ satisfies $\hat{y}=y+\eta$, where $\eta$ is a random noise vector whose every coordinate is distributed as $\NNN(0,\sigma^2)$ for some $\sigma\leq\frac{345600r}{\epsilon|D'|}\sqrt{2k\ln(\frac{dn}{\beta})\ln(\frac{8}{\delta})}$. 
Observe that $\|\eta\|_2^2$ is the sum of the squares of $d$ independent normal random variables $\eta_1,\dots,\eta_d$.
Using tail bounds for the normal distribution, 
provided that $|D|\geq\frac{691200\sqrt{d}}{\epsilon}\ln(\frac{2nd}{\beta})\sqrt{\ln(\frac{8}{\delta})}$, with probability at least $(1-\beta)$ we have that every $|\eta_i|$ is at most $r\sqrt{k/d}$.
Assuming that this is the case, we get that $\|\eta\|_2\leq r\sqrt{k}$, and hence $\|\hat{y}-y\|_2\leq r\sqrt{k}$. Using the triangle inequality, we get that a ball of radius $451r\sqrt{k}$ around the output $\hat{y}$ contains all of the points in $D$, where $|D|\geq t-\frac{216}{\epsilon}\log(\frac{2n}{\beta})$.

Overall, assuming that $t\geq\frac{691416\sqrt{d}}{\epsilon}\log(\frac{6nd}{\beta\delta})\sqrt{\log(\frac{8}{\delta})}$, with probability at least $(1-8\beta)$, the output $\hat{y}$ is s.t.\ at least $t-\frac{216}{\epsilon}\log(\frac{2n}{\beta})$ of the input points are contained inside a ball of radius $451r\sqrt{k}$ around $\hat{y}$.
\end{proof}

Theorem~\ref{thm:main} now follows by combining lemmas~\ref{lem:GoodRadiusUtility}, \ref{lem:GoodRadiusPrivacy}, \ref{lem:GoodCenterPrivacy}, and~\ref{lem:GoodCenterUtility}.

\begin{algorithm*}[!htp]

\caption{\texttt{GoodCenter}}\label{alg:GoodCenter}

{\bf Input:} Database $S\in \R^d$ containing $n$ points, radius $r$, desired number of points $t$, failure probability $\beta$, and privacy parameters $\epsilon,\delta$.
\vspace{-5px}
\begin{enumerate}[leftmargin=15pt,rightmargin=10pt,itemsep=1pt]

\item Let $k=46\log(2n/\beta)$ and let $f:\R^d\rightarrow\R^k$ be a mapping as in Theorem~\ref{thm:JL} (the JL transform).

\item\label{step:abovethreshold} Instantiate algorithm \texttt{AboveThreshold} (Theorem~\ref{thm:AboveThreshold}) with database $S$, privacy parameter $\epsilon/4$, and threshold $t-\frac{100}{\epsilon}\log(2n/\beta)$.

\item\label{step:itr_begin} For every axis $1\leq i\leq k$ of $\R^k$:
\vspace{-5px}
\begin{enumerate}
	\item Choose a random $a_i\in[0,300r]$.
For $j\in\Z$, let $A^i_j$ be the interval $[a_i+j\cdot 300r, \;a_i+(j+1)\cdot 300r )$.
	
\item[\gray{\%}] \gray{
$\{A^i_j\}_{j\in\Z}$ is a partition of the $i^\text{th}$ axis into (randomly shifted) intervals of length $300r$.
}
	
\end{enumerate}

\item For every $\vec{j}=(j_1,\ldots,j_k)\in\Z^k$ let $B_{\vec{j}}\subseteq\R^k$ be the box whose projection on every axis $i$ is $A^i_{j_i}$.

\item\label{step:abovethreshold_1} Issue query $q(S)=\max_{\vec{j}}|f(S)\cap B_{\vec{j}}|$ to \texttt{AboveThreshold}. Denote the received answer as $a$.

\begin{enumerate}[label=\gray{\%},topsep=-10pt]
\item \gray{
That is, $q(S)$ is the maximal number of (projected) input points that are contained within one box.
}
\end{enumerate}

\item\label{step:itr_end} If this step was reached more than $2n\log(1/\beta)/\beta$ times, then halt and fail. Otherwise, if $a=\bot$ then goto Step~\ref{step:itr_begin}.

\begin{enumerate}[label=\gray{\%},topsep=-10pt]
\item \gray{
That is, for at most $2n\log(1/\beta)/\beta$ rounds, we define a partition of $\R^k$ into disjoint rectangles $\{B_{\vec{j}}\}$, and query algorithm \texttt{AboveThreshold} to identify an iteration in which there is a rectangle in $\{B_{\vec{j}}\}$ containing $\gtrsim t$ points.
}
\end{enumerate}
\item\label{step:choosing} Use Theorem~\ref{thm:sanPoints} (stability based techniques) with privacy parameters $(\frac{\epsilon}{4},\frac{\delta}{4})$ to choose a box $B\in\{B_{\vec{j}}\}$ approximately maximizing $|f(S)\cap B|$. Denote $D=\{ x\in S \;:\; f(x)\in B \}$.
\begin{enumerate}[label=\gray{\%},topsep=-10pt]
\item \gray{
That is, $D$ is the set of input points from $S$ that are mapped into $B$ by the mapping $f$.
}
\end{enumerate}


\item\label{step:rotation} Let $Z=(z_1,\dots,z_d)$ be a random orthonormal basis of $\R^d$, and denote $p=900r\sqrt{k\ln(\frac{dn}{\beta})/d}$.

\item For each basis vector $z_i\in Z$:
\vspace{-5px}
\begin{enumerate}

\item\label{step:partitionP} Partition the axis in direction $z_i$ into intervals $\III_i=\{ [j\cdot p\;,\;(j+1)\cdot p ) \; : \; j\in\Z \}$.

\item Define the quality $q(I)$ of every $I\in\III_i$ as the number of points $x\in D$ s.t.\ their projection onto $z_i$ falls in $I$.
\vspace{-5px}
\item\label{step:choosing_1} Use Theorem~\ref{thm:sanPoints} (stability based techniques) with privacy parameters $\left(\frac{\epsilon}{10\sqrt{d \ln(8/\delta)}},\frac{\delta}{8d}\right)$ to choose an interval $I_i\in\III_i$ with large $q(\cdot)$, and let $\hat{I}_i$ denote that chosen interval after extending it by $p$ on each side (that is $\hat{I}_i$ is of length $3p$).

\end{enumerate}

\item\label{step:center_c} Let $c$ be the center of the box in $\R^{d}$ whose projection on every axis $z_i\in Z$ is $\hat{I}_i$, and let $C$ be the ball of radius $2700 r \sqrt{k\ln(\frac{dn}{\beta})}$ around $c$. Define $D'=D\cap C$.
\begin{enumerate}[label=\gray{\%},topsep=-10pt]
\item \gray{
Observe that we defined $D'=D\cap C$, even though we expect that $D\subseteq C$. This will be useful in the privacy analysis, as we now have a deterministic bound on the diameter of $D'$.
}
\end{enumerate}

\item\label{step:output_z} Use the Gaussian mechanism with privacy parameters $(\frac{\epsilon}{4},\frac{\delta}{4})$ to compute and return the noisy average of the points in $D'$ (see Appendix~\ref{sec:NoisyAVG} for details).

\end{enumerate}
\end{algorithm*}

\section{On the impossibility of solving the 1-Cluster problem on infinite domains}\label{sec:lower}

In this section we will show that solving the 1-cluster problem on infinite domains (with reasonable parameters) is impossible under differential privacy. Our lower bound is obtained through a reduction from the simple \emph{interior point problem} defined below.

\begin{definition}
An algorithm $A:X^n \to X$ solves the \emph{interior point problem} on $X$ with error probability $\beta$ if for every $D \in X^n$,
\[\Pr[\min D \le A(D) \le \max D] \ge 1 - \beta,\]
where the probability is taken over the coins of $A$. The sample complexity of the algorithm $A$ is the database size $n$.
\end{definition}

We call a solution $x$ with $\min D \le x \le \max D$ an \emph{interior point} of $D$. Note that $x$ need not be a member of the database $D$.
As was shown in~\cite{BNSV15}, privately solving the interior point problem on a domain $X$ requires sample complexity that grows with $|X|$. In particular, privately solving this problem over an {\em infinite} domain is impossible.

\begin{theorem}[\cite{BNSV15}]\label{thm:range-lb}
Fix any constant $0 < \eps < 1/4$. Let $\delta(n) \le 1/(50 n^2)$. Then for every positive integer $n$, solving the interior point problem on $X$ with probability at least $3/4$ and with $(\eps, \delta(n))$-differential privacy requires sample complexity $n \ge \Omega(\log^*|X|)$.
\end{theorem}

As we will now see, solving the 1-Cluster problem over a domain $X$ (or over $X^d$ for any $d\geq1$) implies solving the interior point problem on $X$.

\begin{theorem}\label{thm:lower}
Let $\beta,\epsilon,\delta,n,t,\Delta,w$ be such that
$\Delta<t\leq n$. If there exists an $(\epsilon,\delta)$-private algorithm that solves the 1-cluster problem $(X,n,t)$ with parameters $(\Delta,w)$ and error probability $\beta$, then there exist a $(2\epsilon,2\delta)$-private algorithm that solves the interior point problem on $X$ with error probability $2\beta$ using sample complexity
$$m= n + 8^{\log^*(4w)} \cdot \frac{144 \log^*(4w)}{\epsilon} \log\Big(\frac{12\log^*(4w)}{\beta\delta}\Big).$$
\end{theorem}

Roughly speaking, Theorem~\ref{thm:lower} states that any solution for the 1-cluster problem (with a reasonable parameter $w$) implies a solution for the interior point problem. 
Let us introduce the following notation for iterated exponentials: 
$$\tower(0)=1 \quad\text{ and }\quad \tower(j)=2^{\tower(j-1)}.$$

\begin{corollary}\label{cor:1clusterLower}
Fix any constants $0 < \epsilon,\beta < 1/8$, and let $\delta(n) \le 1/(200 n^2)$.
Also let $\Delta,t,w$ be s.t.\ $\Delta<t\leq n$ and $w\leq \frac{1}{4}\tower(\log(n^{1/5}/40))$.
For every $(\epsilon,\delta)$-differentially private algorithm that solves the 1-cluster problem $(X,n,t)$ with parameters $(\Delta,w)$ and error probability $\beta$, it holds that $n\geq\Omega(\log^*|X|)$. 
\end{corollary}

Hence, in any solution to the 1-cluster problem where $w$ is smaller than an exponential tower in $n$, the sample size $n$ must grow with $|X|$. In particular, privately solving such 1-cluster problems over infinite domains is impossible.

\begin{algorithm*}[t]
\caption{\texttt{IntPoint}}\label{alg:IntPoint}

\begin{enumerate}[rightmargin=10pt,itemsep=1pt]

\item[] {\bf Input:} Database $S \in X^m$ containing $m$ points from $X$.

\item[] {\bf Algorithm used:} Algorithm $\AAA$ for privately solving the 1-cluster problem $(X,n,t)$ with parameters $(\Delta,w)$ and error probability $\beta$.

\item Sort the entries of $S$ and let $D$ be a multiset containing the middle $n$ entries of $S$.

\item Apply $\AAA$ on $D$ to obtain a center $c\in X$ and a radius $r\in\R$. Let $I\subseteq X$ denote the interval of length $2r$ centered at $c$. If $r=0$ (i.e., $I$ contains only the point $c$), then halt and return $c$.

\begin{enumerate}[label=\gray{\%},topsep=-10pt]
\item \gray{
If $\AAA$ succeeded, then the interval $I$ contains at least $1$ input points from $D$. }
\end{enumerate}

\item Partition $I$ into intervals of length $r/w$, and let $\JJJ$ be the set of all edge points of the resulting intervals.

\begin{enumerate}[label=\gray{\%},topsep=-10pt]
\item \gray{
As $I$ contains points from $D$, at least one of the above intervals contains points from $D$. We will show that this interval cannot contain all of $D$, and hence, one of its two edge points is an interior point of $D$ (and of $S$).}
\end{enumerate}

\item Define $q:X^*\times X\rightarrow\N$ as $q\big((x_1,\dots,x_n),a\big)=\min\big\{ \;|\{i : x_i\leq a\}| \;,\;  |\{i : x_i\geq a\}| \; \big\}$. Apply algorithm \texttt{RecConcave} on the database $S$ with privacy parameters $(\epsilon,\delta)$, utility parameters  $(\alpha{=}\frac{1}{2},\beta)$, quality function $q$, and quality promise $\frac{m-n}{2}$ to choose and return $j\in\JJJ$.

\begin{enumerate}[label=\gray{\%},topsep=-10pt]
\item \gray{
If a point $j^*\in\JJJ$ is an interior point of $D$, then it is also an interior point of $S$ with quality $q(S,j^*)\geq \frac{m-n}{2}$ (since $D$ contains the middle $n$ elements of $S$). Hence, w.h.p., algorithm \texttt{RecConcave} identifies an interior point of $S$.}
\end{enumerate}

\end{enumerate}

\end{algorithm*}

\begin{proof}[Proof of Corollary~\ref{cor:1clusterLower}]
Denote $\delta(n) = 1/(200 n^2)$, and let $\AAA$ be a $(\epsilon{=}\frac{1}{10},\delta(n))$-differentially private algorithm that solves the 1-cluster problem $(X,n,t)$ with parameters $(\Delta,w)$ and error probability $\beta{=}\frac{1}{10}$, where $\Delta<t\leq n$ and $w\leq \frac{1}{4}\tower(\log(n^{1/5}/40))$. By Theorem~\ref{thm:lower}, there exist a $(2\epsilon{=}\frac{1}{5},2\delta(n))$-differentially private algorithm that solves the interior point problem on $X$ with error probability $2\beta{=}\frac{1}{5}$ using sample complexity
$$m= n + 8^{\log^*(4w)} \cdot 1440 \log^*(4w) \log\Big(24000 n^2 \log^*(4w)\Big).$$
Using the assumption that $w\leq \frac{1}{4}\tower(\log(n^{1/5}/40))$, we get that $m\leq2n$, and therefore $2\delta(n)=\frac{1}{100 n^2}\leq\frac{1}{50 m^2}$. Theorem~\ref{thm:range-lb} now states that $m \ge \Omega(\log^*|X|)$, and hence, $n \ge \Omega(\log^*|X|)$.
\end{proof}

\begin{proof}[Proof of Theorem~\ref{thm:lower}]
The proof is via the construction of algorithm \texttt{IntPoint}. For the privacy analysis, observe that on step~1 the algorithm constructs a database $D$ containing the $n$ middle elements of the input database $S$. Fix two neighboring databases $S_1,S_2$ of size $m$, and consider the databases $D_1,D_2$ containing the middle $n$ elements of $S_1,S_2$ respectively. As $S_1,S_2$ differ in at most one element, so does $D_1,D_2$ (our algorithms ignore the order of their input database). Hence, applying a private computation onto $D$ preserves privacy. Algorithm \texttt{IntPoint} interacts with $D$ using algorithm $\AAA$ (on step~1), and interacts with $S$ using algorithm \texttt{RecConcave} (on step~4). Hence, by composition (see Theorem~\ref{thm:composition1}), Algorithm \texttt{IntPoint} is $(2\epsilon,2\delta)$-differentially private.

As for the utility analysis, let $S\in X^m$ be an instance to the interior point problem, and consider the execution of algorithm \texttt{IntPoint} on $S$. Let $D$ be the multiset defined on step~1 of the execution, and let $c,r,I$ be the center, the radius, and the interval obtained on step~2. By the properties of algorithm $\AAA$, with probability $1-\beta$, the interval $I$ (of length $2r$ centered at $c$) contains at least $t-\Delta\geq1$ input points from $D$, and moreover, $2r\leq2w\cdot r_{opt}$ where $2r_{opt}$ is the length of the smallest interval containing $t$ points from $D$. That is, {\em every} interval containing $t$ points from $D$ is of  length at least $2r/w$. We continue with the analysis assuming that this is the case. 

If $r=0$, then the interval $I$ contains only the point $c$, and $c$ is therefore an interior point of $D$ (and of $S$). We hence proceed with the analysis assuming that $r>0$.

On step~3 we partition the interval $I$ (of length $2r$) into intervals of length $r/w$. Let us denote them as $\III=\{I_i\}$. Since $I$ contains at least 1 point from $D$, there exists an interval $I_\ell\in\III$ that contains a point from $D$. In addition,
$I_\ell$ is of length $r/w$, and can hence contain at most $t-1$ points from $D$. So, there exists an interval $I_\ell\in\III$ containing some of the points in $D$, but not all of them. One of its edge points must be, therefore, an interior point of $D$. Thus, the set $\JJJ$ of all end points of the intervals in $\III$ (defined on step~3) contains an interior point of $D$.

Recall that the input database $S$ contains at least $\frac{m-n}{2}$ elements which are bigger than the elements in $D$ and $\frac{m-n}{2}$ elements which are smaller. Hence, there exists a point $j^*\in\JJJ$ s.t. $q(S,j^*)\geq \frac{m-n}{2}$, and the quality promise given to algorithm \texttt{RecConcave} is valid. Observe that $|\JJJ|\leq4w$ (this is the size of the solution set given to algorithm \texttt{RecConcave}). Hence, assuming that
$$
m\geq n + 8^{\log^*(4w)} \cdot \frac{144 \log^*(4w)}{\epsilon} \log\Big(\frac{12\log^*(4w)}{\beta\delta}\Big),
$$
with probability at least $1-\beta$, the output $j\in\JJJ$ is s.t.\ $q(S,j)\geq \frac{m-n}{4}$, and $j$ is an interior point of $S$ as required.
\end{proof}

\section{Sample and aggregate}\label{sec:sa}

Consider $f:U^*\rightarrow X^d$ mapping databases to $X^d$. Fix a database $S\in U^*$, and assume that evaluating $f$ on a random sub-sample $S'$ (containing iid samples from $S$) results in a good approximation to $f(S)$. Our goal is to design a private analogue to $f$.

\begin{definition}
Fix a function $f:U^*\rightarrow X^d$ and a database $S\in U^*$. A point $c\in X^d$ is an {\em $(m,r,\alpha)$-stable point} of $f$ on $S$ if for a database $S'$ containing $m$ iid samples from $S$ we have  $\Pr[\|f(S')-c\|_2\leq r]\geq\alpha$.
If such a point $c$ exists,we say that  $f$ is $(m,r,\alpha)$-stable on $S$. We will call $r$ the {\em radius} of the stable point $c$.
\end{definition}

In the sample and aggregate framework~\cite{NRS07}, the goal is to privately identify an $(m,r,\alpha)$-stable point of a function $f$ on the given input database $S$.
Intuitively, if $r$ is small and $\alpha$ is large, then such a point would be a good (private) substitute for (the non-private value) $f(S)$. Note that a stable point with small $r$ and with large $\alpha$ might not exist.



\begin{theorem}[\cite{NRS07} restated, informal]\label{thm:NRS07}
Fix a desired stability parameter $m$, and let $n\geq 4d^2 m$.
There exists an efficient differentially private algorithm that given a database $S\in U^n$ and a function $f:U^*\rightarrow X^d$, identifies an $\Big(m,O(\sqrt{d}\cdot r_{opt})+\sqrt{d}\cdot e^{-\Omega\left(\sqrt{\frac{n}{m}}\frac{1}{d}\right)},0.51\Big)$-stable point of $f$ on $S$, where $r_{opt}$ is the smallest $r$ s.t. $f$ is $(m,r,0.51)$-stable on $S$.
\end{theorem}

Note the following caveats in Theorem~\ref{thm:NRS07}: 
(1)~The function $f$ might only be $\left(m,r,0.51\right)$-stable on $S$ for very large values of $r$, in which case the similarity to (the non-private) $f(S)$ is lost. (2)~The error in the radius of the stable point grows with $\sqrt{d}$, which might be unacceptable in high dimensions.\footnote{
Let $c$ be an $(m,r,0.51)$-stable point of $f$ on $S$. In~\cite{NRS07}, the returned point $c'$ is s.t.\ that the error vector $(c-c')$ has magnitude $O(r_{opt})+e^{-\Omega\left(\sqrt{\frac{n}{m}}\frac{1}{d}\right)}$ in each coordinate.}
Using Theorem~\ref{thm:main} it is possible to avoid those two caveats.

\begin{theorem}\label{thm:SA}
Fix privacy parameters $\epsilon,\delta$, failure parameter $\beta$, and desired stability parameters $m,\alpha$. Let\\
$$\frac{n}{m}\geq O\left(\frac{\sqrt{d}\cdot \log\left(\frac{1}{\beta}\right)}{\alpha\cdot\min\{\alpha,\epsilon\}}\log\left(\frac{nd}{\alpha\beta\epsilon\delta}\right)\sqrt{\log\left(\frac{1}{\alpha\beta\epsilon\delta}\right)}\cdot 9^{\log^*(2|X|\sqrt{d})}\right).$$
There exists an efficient $(\epsilon,\delta)$-differentially private algorithm that given a database $S\in U^n$ and a function $f:U^*\rightarrow X^d$, with probability at least $(1-\beta)$, identifies an $(m,O(r_{opt}\cdot\sqrt{\log n}),\frac{\alpha}{8})$-stable point of $f$ on $S$, where $r_{opt}$ is the smallest $r$ s.t. $f$ is $(m,r,\alpha)$-stable on $S$.
\end{theorem}

Theorem~\ref{thm:SA} is proved using algorithm \texttt{SA}.
Similarly to~\cite{NRS07}, the idea is to apply $f$ onto $k$ random subsamples of the input database $S$ (obtaining outputs $Y=\{y_1,y_2,\dots,y_k\}$ in $X^d$), and then to privately identify a point $z\in X^d$ which is close to points in $Y$. We will use the following lemma to argue that the random subsampling step maintains privacy: 

\begin{lemma}[\cite{KLNRS08,BNSV15}]\label{lem:iidSampling}
Fix $\epsilon\leq1$ and let $\cal A$ be an $(\eps, \delta)$-differentially private algorithm operating on databases of size $m$.
For $n\geq2m$, construct an algorithm $\tilde{\cal A}$ that on input a database $D$ of size $n$ subsamples (with replacement) $m$ rows from $D$ and runs $\cal A$ on the result. Then $\tilde{\cal A}$ is $( \tilde{\eps} , \tilde{\delta} )$-differentially private for
$\tilde{\eps}=6\eps m /n$ and $\tilde{\delta}= \exp(6\eps m/n)\frac{4m}{n}\cdot\delta.$
\end{lemma}

\begin{lemma}\label{lem:SAPrivacy}
Algorithm \texttt{SA} is $(\epsilon,\delta)$-differentially private.
\end{lemma}

\begin{proof}
Let $\cal A$ denote an algorithm identical to \texttt{SA}, except without the iid sampling on step~\ref{step:SA_sample} (the input to $\cal A$ is the database $D$), and observe that $\cal A$ preserve $(\epsilon,\delta)$-differential privacy. To see this, let $D,D'$ be two neighboring databases, and consider the execution of $\cal A$ on $D$ and $D'$.
Next, note that there is at most one index $i$ s.t.\ $D_i$ differs from $D'_i$. Hence, $Y$ and $Y'$ are neighboring databases, and privacy is preserved by the properties of algorithm $\cal M$.

Algorithm \texttt{SA} (including the iid sampling) is $(\epsilon,\delta)$-differentially private by Lemma~\ref{lem:iidSampling} (iid sampling).
\end{proof}

\begin{algorithm*}[t]
\caption{\texttt{SA}}

\begin{enumerate}[rightmargin=10pt,itemsep=1pt]

\item[] {\bf Input:} Database $S$ containing $n$ elements from $U$, function $f:U^*\rightarrow X^d$, privacy parameters $\epsilon,\delta$, failure parameter $\beta$, and desired stability parameters $m,\alpha$. We denote $k\triangleq\frac{n}{9m}$.

\item[] {\bf Algorithm used:} An $(\epsilon\leq\frac{\alpha}{72},\delta\leq\frac{\beta\epsilon}{3})$-private algorithm $\cal M$ for solving the 1-cluster problem $(X^d,k,t)$ for every $t\geq t_{\rm min}$, with parameters $(\Delta{\leq}t/2,w)$ and error probability $\frac{\beta}{3}$.


\item\label{step:SA_sample} Let $D$ be the outcome of $n/9$ iid samples from $S$. Partition $D$ into $k$ databases $D_1,D_2,...,D_k$ of size $m$ each.

\item\label{step:SA_Y} Let $Y=\{y_1=f(D_1),y_2=f(D_2),\dots,y_k=f(D_k)\}$.

\item Apply $\cal M$ on the database $Y$ with parameter $t=\frac{\alpha k}{2}$ to get a point $z$. Output $z$.

\end{enumerate}
\end{algorithm*}

In the utility analysis of algorithm \texttt{SA} we first show that a ball around the returned point $z\in X^d$ contains $\gtrsim \alpha$ fraction of points in $Y=\{y_1=f(D_1),y_2=f(D_2),\dots,y_k=f(D_k)\}$. Afterwards, we will argue that this ball also contains $\gtrsim \alpha$ mass of the underlying distribution (i.e., the distribution defined by applying $f$ on a random subsample of $S$). The straightforward approach for such an argument would be to use VC bounds stating that for any ball (in particular, the ball around $z$) it holds that the fraction of points from $Y$ in it is close to the weight of the ball w.r.t.\ the underlying distribution. However, for this argument to go through we would need $|Y|$ to be as big as the VC dimension of the class of $d$-dimensional balls, which is $d+1$. This seems wasteful since our private algorithm for locating the ball only requires the generation of a small cluster of size $\approx\sqrt{d}$. Instead, will make use of the inherent generalization properties of differential privacy, first proven by Dwork et al.~\cite{DFHPRR14}. Specifically, we will use the following theorem of Bassily et al.~\cite{BNSSSU15} stating that any predicate computed with differential privacy automatically provides generalization:

\begin{theorem}[\cite{BNSSSU15}] \label{thm:DPgeneralization}
Let $\epsilon \in (0,1/3)$, $\delta \in (0,\epsilon/4)$, and $n\geq\frac{1}{\epsilon^2}\log(\frac{4\epsilon}{\delta})$.
Let $\AAA:U^n\rightarrow2^U$ be an $(\epsilon,\delta)$-differentially private algorithm that operates on a database of size $n$ and outputs a predicate $h:U\rightarrow\{0,1\}$.
Let $\DDD$ be a distribution over $U$, let $S$ be a database containing $n$ i.i.d.\ elements from $\DDD$, and let $h\leftarrow \AAA(S)$.
Then,
$$
\Pr_{S,\AAA}\left[|h(S)-h(\DDD)| > 18\epsilon\right]\leq \frac{\delta}{\eps},
$$
where $h(S)$ is the empirical average of $h$ on $S$, and $h(\DDD)$ is the expectation of $h$ over $\DDD$.
\end{theorem}

\begin{lemma}\label{lem:SAutility}
Let \texttt{SA} be executed on a function $f$ and on a database $S$ of size $n$ such that $f$ is $(m,r,\alpha)$-stable on $S$. Assume \texttt{SA} has access to an $(\epsilon\leq\frac{\alpha}{72},\delta\leq\frac{\beta\epsilon}{3})$-private algorithm $\cal M$ for solving the 1-cluster problem $(X^d,k,t)$ for every $t\geq t_{\rm min}$ with parameters $(\Delta{\leq}t/2,w)$ and error probability $\frac{\beta}{3}$.
With probability at least $(1-\beta)$, the output $z$ is an $(m,wr,\frac{\alpha}{8})$-stable point of $f$ on $S$, provided that 
$$n\geq m\cdot O\left(\frac{t_{\rm min}}{\alpha}+\frac{1}{\alpha^2}\log\left(\frac{12}{\beta}\right)\right).$$
\end{lemma}

\begin{proof}
Let algorithm \texttt{SA} be executed on a function $f$ and on a database $S$ such that $f$ is $(m,r,\alpha)$-stable on $S$. Let $c\in X^d$ be a stable point of $f$ on $S$, and let $\BBB_r(c,Y)$ denote the number of points in $Y$ within distance $r$ from $c$.

By the definition of the stable point $c$, for every $y_i$ (defined on step~\ref{step:SA_Y}) we have that $\Pr[\|y_i-c\|_2\leq r]\geq\alpha$. Hence, by the Chernoff bound, $\Pr[\BBB_r(c,Y)<\frac{\alpha k}{2}]\leq\exp(-\frac{\alpha k}{8})\leq\beta/3$. 

Assume that $\BBB_r(c,Y)\geq\frac{\alpha k}{2}$, i.e., there are at least $\frac{\alpha k}{2}=t$ points in $Y$ within distance $r$ from $c$. Hence, by the properties of algorithm $\cal M$, with probability at least $(1-\beta/3)$ the output $z$ is s.t.\ $\BBB_{wr}(z,Y)\geq \frac{t}{2}=\frac{\alpha k}{4}$.

We now argue that the ball of radius $wr$ around $z$ not only contains a lot of points from $Y$, but is also ``heavy'' w.r.t.\ the underlying distribution (i.e., the distribution defined by applying $f$ on a random subsample of $S$). To that end, consider the predicate $h:X^d\rightarrow\{0,1\}$ defined as $h(x)=1$ iff $\|x-z\|_2\leq wr$. That is, $h$ evaluates to 1 exactly on points inside the ball of radius $wr$ around $z$. So $h(Y)\geq\frac{\alpha}{4}$. By Theorem~\ref{thm:DPgeneralization}, assuming that $k\geq\frac{5184}{\alpha^2}\log(\frac{12}{\beta})$, with probability at least $1-\delta/\epsilon\geq1-\beta/3$ we have that for a random subsample $S'$ containing $m$ i.i.d.\ samples from $S$,  it holds that $\Pr[h(f(S'))=1]=\Pr[\|f(S')-z\|_2\leq wr]\geq\frac{\alpha}{4}-18\epsilon\geq\frac{\alpha}{8}$. 

All in all, provided that $n=9mk\geq\frac{18m}{\alpha} t_{\rm min}+\frac{46646m}{\alpha^2}\log(\frac{12}{\beta})$, with probability at least $(1-\beta)$, the output $z$ is an $(m,wr,\frac{\alpha}{8})$-stable point of $f$ on $S$.
\end{proof}

Theorem~\ref{thm:SA} now follows from combining lemmas~\ref{lem:SAPrivacy}, and~\ref{lem:SAutility} with Theorem~\ref{thm:main}.

\bibliographystyle{abbrv}

\appendix

\section{Noisy average of vectors in $\R^d$}\label{sec:NoisyAVG}

{
\renewcommand{\thetheorem}{\ref{thm:gauss}}
\begin{theorem}[The Gaussian Mechanism \cite{DKMMN06}]
Let $\epsilon,\delta\in(0,1)$, and let $f:X^* \rightarrow \R^d$ be a function of $L_2$-sensitivity $k$. Denote $\sigma\geq\frac{k}{\epsilon}\sqrt{2\ln(1.25/\delta)}$.
The mechanism $\AAA$ that on input $D\in X^*$ 
adds independently generated noise with distribution $\NNN(0,\sigma^2)$ to each of the $d$ output terms of $f(D)$ preserves $(\epsilon,\delta)$-differential privacy.
\end{theorem}
\addtocounter{theorem}{-1}
}

We will use the Gaussian mechanism to obtain (noisy) averages of vectors in $\R^d$: 
Let $g:\R^d\rightarrow\{0,1\}$ be a predicate over vectors in $\R^d$, and denote 
$$\Delta_g=\max_{
\begin{array}{c} v\in\R^d \text{ s.t.}\\ g(v)=1 \end{array}}\|v\|_2.$$
Given a multiset of vectors $V$, we are interested in approximating $\frac{\sum_{v\in V:g(v)=1}v}{|\{v\in V:g(v)=1\}|}$. Since both the numerator and the denominator have bounded $L_2$-sensitivity, we could estimate each of them separately using the Gaussian mechanism.
Alternatively, we could directly analyze the $L_2$-sensitivity of $g(V)\triangleq\frac{\sum_{v\in V:g(v)=1}v}{|\{v\in V:g(v)=1\}|}$ and apply the Gaussian mechanism directly to $g(V)$. This has the advantage of having only an additive Gaussian error, rather than a noisy numerator/denominator. It is easier to analyze.

We now bound the $L_2$ sensitivity of $g(V)$.
Let $V$ be a multiset of vectors in $\R^d$ and let $V'=V\cup\{u\}$ for some vector $u$.
If $g(u)=0$ then $g(V)=g(V')$ and $\|g(V)-g(V')\|_2=0$. We assume therefore that $g(u)=1$.

Let $m=|\{v\in V:g(v)=1\}|$.
We need to analyze the $L_2$ norm of the vector $\nu=\frac{1}{m}\left(\sum_{v\in V : g(v)=1}v\right) - \frac{1}{m+1}\left(u+\sum_{v\in V : g(v)=1}v\right)$.
To that end, observe that the $i^\text{th}$ coordinate of $\nu$ is
\begin{eqnarray*}
\nu_i &=& \frac{1}{m}\left(\sum_{v\in V : g(v)=1}v_i\right) - \frac{1}{m+1}\left(u_i+\sum_{v\in V : g(v)=1}v_i\right)\\
&=& \frac{\left(\sum_{v\in V : g(v)=1}v_i\right) - mu_i}{m(m+1)}\\
&=& \frac{g(V)_i - u_i}{m+1}.
\end{eqnarray*}
Hence,
$$
\|\nu\|_2 = \sqrt{(\nu_1)^2+\cdots(\nu_d)^2} 
=\frac{1}{m+1}\sqrt{\langle u,u \rangle - 2\langle u,g(V) \rangle + \langle g(V),g(V) \rangle }
\leq \frac{2\Delta_g}{m+1},
$$
where the last inequality is since 
\begin{eqnarray*}
\left| \langle u,g(V) \rangle \right|&=&\left| u_1\cdot g(V)_1 +\cdots+u_d\cdot g(V)_d \right|\\
&=&\left|  \frac{u_1}{m}\left(\sum_{v\in V : g(v)=1}v_1\right) +\cdots+ \frac{u_d}{m}\left(\sum_{v\in V : g(v)=1}v_d\right) \right|\\
&\leq& \frac{1}{m}\sum_{v\in V : g(v)=1}\left|\langle u,v \rangle\right|\\
&\leq& \frac{1}{m}\sum_{v\in V : g(v)=1}\|u\|_2\cdot \|v\|_2\leq(\Delta_g)^2,
\end{eqnarray*}
and similarly
$$
\left| \langle g(V),g(V) \rangle \right| = \left| \left\langle \frac{1}{m}\sum_{v\in V : g(v)=1}v ,g(V) \right\rangle \right|
= \frac{1}{m} \left|  \sum_{v\in V : g(v)=1} \langle v ,g(V) \rangle \right|\leq(\Delta_g)^2.
$$

So, for any $V$ and $V'=V\cup\{u\}$ we have that $\|g(V)-g(V')\|_2\leq \frac{2\Delta_g}{m+1}$, where $m=|\{v\in V:g(v)=1\}|$. Therefore, for any two neighboring sets $V_1=V\cup\{u\}$ and $V_2=V\cup\{v\}$ we have that $\|g(V_1)-g(V_2)\|_2\leq \frac{4\Delta_g}{m+1}$ by the triangle inequality.

We will use algorithm \texttt{NoisyAVG} (algorithm~\ref{alg:NoisyAVG}) for privately obtaining (noisy) averages of vectors in $\R^d$.

\begin{observation}
Let $V$ and $g$ be s.t.\ $m=|\{v\in V : g(v)=1\}|\geq\frac{16}{\epsilon}\ln(\frac{2}{\beta\delta})$. 
With probability at least $(1-\beta)$ algorithm $\texttt{NoisyAVG}(V)$ returns $g(V)+\eta$ where $\eta$ is a vector whose every coordinate is sampled i.i.d.\ from $\NNN(0,\sigma^2)$ for some $\sigma\leq\frac{16\Delta_g}{\epsilon m}\sqrt{2\ln(8/\delta)}$.
\end{observation}

\begin{observation}
The requirement that $\max_{v: g(v)=1 }\|v\|_2\leq\Delta_g$ could easily be replaced by
$$\max_{u,v: g(u)=g(v)=1}\|u-v\|_2\leq\Delta_g.$$
That is, the predicate $g$ may define a subset of $\R^d$ of diameter $\leq\Delta_g$, not necessarily around the origin.
\end{observation}

\begin{proof}
Let $c\in\R^d$ be s.t.\ $g(c)=1$ and $\|c-v\|_2\leq\Delta_g$ for every $v\in\R^d$ s.t.\ $g(v)=1$.
Given a vector set $V$, we could apply algorithm \texttt{NoisyAVG} on the set $\tilde{V}=\{v-c:v\in V\}$ with the predicate $\tilde{g}:\R^d\rightarrow\{0,1\}$ s.t.\ $\tilde{g}(v)=g(v+c)$, and add the vector $c$ to the result.
\end{proof}

\begin{algorithm*}[t]
\caption{\texttt{NoisyAVG}}\label{alg:NoisyAVG}
{\bf Input:} Multiset $V$ of vectors in $\R^d$, predicate $g$, parameters $\epsilon,\delta$.
\begin{enumerate}[rightmargin=10pt,itemsep=1pt]
\item Set $\hat{m} = |\{v\in V : g(v)=1\}| + \Lap(2/\epsilon) - \frac{2}{\epsilon}\ln(2/\delta)$. If $\hat{m}\leq 0$ then output $\bot$ and halt.
\item Denote $\sigma = \frac{8\Delta_g}{\epsilon\hat{m}}\sqrt{2\ln(8/\delta)}$, and let $\eta\in\R^d$ be a random noise vector with each coordinate sampled independently from $\NNN(0,\sigma^2)$. Return $g(V)+\eta$.
\end{enumerate}
\end{algorithm*}

\begin{theorem}\label{thm:NoisyAVG}
Algorithm \texttt{NoisyAVG} is $(\epsilon,\delta)$-differentially private.
\end{theorem}

\begin{proof}
For a given parameter $\sigma$, let us denote by $\eta(\sigma)$ a random vector whose every coordinate is sampled iid from $\NNN(0,\sigma^2)$.
Now fix two neighboring sets $V,V'$ and a predicate $g$. Denote $m=|\{v\in V:g(v)=1\}|$, and recall that $\|g(V)-g(V')\|_2\leq\frac{4\Delta_g}{m}$.
The standard analysis of the Gaussian mechanism (see, e.g.,~\cite{DR14}) shows that for any $F\subseteq\R^d$ and for any $\sigma\geq\frac{8\Delta_g}{\epsilon m}\sqrt{2\ln(8/\delta)}$ it holds that
$$
\Pr_{\eta(\sigma)}[g(V)+\eta(\sigma)\in F]\leq e^{\epsilon/2} \cdot \Pr_{\eta(\sigma)}[g(V')+\eta(\sigma)\in F]+\frac{\delta}{6}.
$$

We will use $\hat{m}(V)$ to denote $\hat{m}$ as it is on step~1 of the execution of \texttt{NoisyAVG} on $V$. 
Note that $\Pr[\hat{m}(V)>m]=\Pr[\Lap(\frac{2}{\epsilon})>\frac{2}{\epsilon}\ln(2/\delta)]\leq\frac{\delta}{2}$. Hence, for any set of outputs $F$ s.t.\ $\bot\notin F$ we have that
\begin{eqnarray*}
&&\Pr[\texttt{NoisyAVG}(V)\in F]\\
&&\qquad\qquad = \int_{0}^{\infty} \Pr[\hat{m}(V)=\tilde{m}] \cdot \Pr[\texttt{NoisyAVG}(V)\in F | \hat{m}(V)=\tilde{m}] d\tilde{m}\\
&&\qquad\qquad\leq \frac{\delta}{2} + \int_{0}^{m} \Pr[\hat{m}(V)=\tilde{m}] \cdot \Pr[\texttt{NoisyAVG}(V)\in F | \hat{m}(V)=\tilde{m}] d\tilde{m}\\
&&\qquad\qquad= \frac{\delta}{2} +  \int_{0}^{m} \Pr[\hat{m}(V)=\tilde{m}] \cdot \Pr\left[g(V)+\eta\left(\frac{8\sqrt{2}}{\epsilon\tilde{m}}\sqrt{\ln(\frac{8}{\delta})}\right) \in F\right] d\tilde{m}\\
&&\qquad\qquad\leq \frac{\delta}{2} + \int_{0}^{m} e^{\epsilon/2}\cdot \Pr[\hat{m}(V')=\tilde{m}]\cdot\left(e^{\epsilon/2}\Pr\left[g(V')+\eta\left(\frac{8\sqrt{2}}{\epsilon\tilde{m}}\sqrt{\ln(\frac{8}{\delta})}\right) \in F\right]+\frac{\delta}{6}\right) d\tilde{m}\\
&&\qquad\qquad\leq \frac{\delta}{2} + \int_{0}^{\infty} e^{\epsilon/2}\cdot \Pr[\hat{m}(V')=\tilde{m}]\cdot\left(e^{\epsilon/2}\Pr\left[\texttt{NoisyAVG}(V')\in F | \hat{m}(V')=\tilde{m}\right]+\frac{\delta}{6}\right) d\tilde{m}\\
&&\qquad\qquad\leq \delta+ e^\epsilon\cdot\Pr[\texttt{NoisyAVG}(V')\in F].
\end{eqnarray*}

For a set of outputs $F$ s.t.\ $\bot\in F$ we have
\begin{eqnarray*}
\Pr[\texttt{NoisyAVG}(V)\in F] &=& \Pr[\texttt{NoisyAVG}(V)=\bot] + \Pr[\texttt{NoisyAVG}(V)\in F\setminus\{\bot\}]\\
&\leq& e^\epsilon\Pr[\texttt{NoisyAVG}(V')=\bot] + e^\epsilon\Pr[\texttt{NoisyAVG}(V')\in F\setminus\{\bot\}]+\delta\\
&=& e^\epsilon\Pr[\texttt{NoisyAVG}(V')\in F]+\delta.
\end{eqnarray*}
\end{proof}

\end{document}